\documentclass[a4paper,11pt]{article}
\usepackage{amsmath,amssymb}
\usepackage{bm,dsfont,mathtools,stmaryrd}
\usepackage{xparse,xstring,xspace}
\usepackage{amsthm,thmtools}
\usepackage{mathtools}
\usepackage[top=80pt, bottom=90pt, left=70pt, right=70pt]{geometry}
\usepackage{upref}
\usepackage{comment}
\usepackage{enumitem}
\usepackage[dvipdfmx]{graphicx}
\usepackage[%
    bookmarksnumbered,
    hypertexnames=false,
    pdfdisplaydoctitle,
    pdfusetitle,
    unicode,
    colorlinks=true,
    linkcolor=blue,
    citecolor=blue
]{hyperref}
\usepackage{color}
\usepackage[capitalise,noabbrev]{cleveref}

\crefname{line}{Line}{Lines}
\Crefname{line}{Line}{Lines}

\usepackage{autonum}
\makeatletter
\autonum@generatePatchedReferenceCSL{Cref}
\makeatother

\usepackage{mleftright}
\mleftright

\usepackage{algorithm}
\usepackage{algpseudocode}
\usepackage{bbm}
\usepackage{setspace}

\declaretheorem[style=plain,numberwithin=section,name=Theorem]{theorem}
\declaretheorem[style=plain,sibling=theorem,name=Lemma]{lemma}
\declaretheorem[style=plain,sibling=theorem,name=Proposition]{proposition}

\declaretheorem[style=plain,sibling=theorem,name=Claim]{claim}

\crefname{theorem}{Theorem}{Theorems}
\crefname{proposition}{Proposition}{Propositions}
\crefname{lemma}{Lemma}{Lemmas}
\crefname{exmp}{Example}{Examples}
\crefname{corollary}{Corollary}{Corollaries}
\crefname{claim}{Claim}{Claims}
\crefname{remark}{Remark}{Remarks}
\crefname{section}{Section}{Sections}

\newcommand{\Z}{\mathbb{Z}}
\newcommand{\R}{\mathbb{R}}

\newcommand{\cI}{\mathcal{I}}
\newcommand{\cR}{\mathcal{R}}
\newcommand{\cS}{\mathcal{S}}
\newcommand{\cL}{\mathcal{L}}
\newcommand{\cE}{\mathcal{E}}
\newcommand{\dcov}{d_\mathrm{cov}}
\newcommand{\dsum}{d_\mathrm{sum}}
\newcommand{\ddiv}{d_\mathrm{div}}
\newcommand{\SumkDiverse}{{\sc Sum-$k$-Diverse}\xspace}
\newcommand{\CovkDiverse}{{\sc Cov-$k$-Diverse}\xspace}
\newcommand{\Mincut}{{\sc Unweighted Minimum\ $s$-$t$ Cut}\xspace}
\newcommand{\StableMatching}{{\sc Stable Matching}\xspace}

\newcommand{\condS}{{\rm S}}
\newcommand{\condR}{{\rm R}}
\newcommand{\condT}{{\rm T}}

\title{
A general framework for finding diverse solutions via network flow\\
and its applications
}
\author{Yuni Iwamasa\thanks{Graduate School of Informatics, Kyoto University, Kyoto 606-8501, Japan. Email: \texttt{iwamasa@i.kyoto-u.ac.jp}, \texttt{morihira.shunya.27e@st.kyoto-u.ac.jp}} 
\and 
Tomoki Matsuda\thanks{School of Computing, Institute of Science Tokyo, Tokyo 152-8550, Japan. Email: \texttt{sumita@comp.isct.ac.jp}} 
\and 
Shunya Morihira\footnotemark[1]
\and 
Hanna Sumita\footnotemark[2]}

\date{}

\begin{document}

\maketitle

\begin{abstract}
In this paper, we present a general framework for efficiently computing diverse solutions to combinatorial optimization problems.
Given a problem instance, the goal is to find $k$ solutions that maximize a specified diversity measure; the sum of pairwise Hamming distances or the size of the union of the $k$ solutions.
Our framework applies to problems satisfying two structural properties: (i) All solutions are of equal size and (ii) the family of all solutions can be represented by a surjection from the family of ideals of some finite poset.
Under these conditions, we show that the problem of computing $k$ diverse solutions can be reduced to the minimum cost flow problem and the maximum $s$-$t$ flow problem.
As applications, we demonstrate that both the unweighted minimum $s$-$t$ cut problem and the stable matching problem satisfy the requirements of our framework.
By utilizing the recent advances in network flows algorithms, we improve the previously known time complexities of the diverse problems, which were based on submodular function minimization.
\end{abstract}

\section{Introduction}

Modeling real-world problems as optimization problems requires balancing accuracy and tractability, and many factors are inevitably omitted from the model.
Consequently, finding a single optimal solution to the optimization problem may not be enough because the solution may not be suitable for the real-world problem.
Motivated by this issue, developing algorithms to find \emph{diverse} multiple solutions has recently gain more attention in the field of combinatorial optimization~\cite{baste2022, baste2019,deBerg2023,deBerg2025, eiben2024, fomin2024a, fomin2024b, Hanaka2023, hanaka2021}. 

We formally define the \emph{$k$-diverse problem} for a positive integer $k$,
which is our target problem.
Let $\ddiv$ be a function that measures the ``diversity'' for a $k$-tuple of subsets of some arbitrary fixed ground set $E$.
For a combinatorial problem {\sc Prob},
the \emph{$k$-diverse problem} of {\sc Prob} asks, given an instance $\mathbf{I}$ of {\sc Prob},
to find a $k$-tuple $(S_1, S_2, \dots, S_k)$ of solutions of $\mathbf{I}$
that maximizes the fixed measure $\ddiv$.
Typical examples of the measure of diversity include the following:
\begin{align}
    \dsum(S_1, S_2, \dots, S_k)  \coloneqq \sum_{1\leq i<j\leq k}\left|S_i\triangle S_j\right|, \qquad
    \dcov(S_1, S_2, \dots, S_k)  \coloneqq \left|\bigcup_{1\leq i\leq k}S_i\right|.
\end{align}
If we use $\dsum$ (resp.~$\dcov$) as the measure,
then we refer to the corresponding $k$-diverse problem as
\SumkDiverse (resp.~\CovkDiverse).

In this paper, we present a general framework for efficiently finding diverse solutions to combinatorial optimization problems.
In particular, we focus on \SumkDiverse/\CovkDiverse of a combinatorial problem {\sc Prob} having two properties~\eqref{cond:S} and~\eqref{cond:R} introduced below.
The first property~\eqref{cond:S} states that all solutions of each instance of {\sc Prob} have the same size:
\begin{enumerate}[label=(\condS),ref=\condS]
\item \label{cond:S}
For any instance $\mathbf{I}$ of {\sc Prob}, there exist a finite ground set $E$ and a positive integer $q$ such that
the family $\cS(\mathbf{I})$ of solutions of $\mathbf{I}$ consists of subsets of $E$ with the same size $q$.
\end{enumerate}
Suppose that the combinatorial problem {\sc Prob} has the property~\eqref{cond:S}.
Let $(P, \preceq)$ be a finite poset with minimum element $\bot$ and maximum element $\top$.
For such a poset $P$, let $P^*$ denote the poset obtained from $P$ by removing $\bot$ and $\top$.
For a map $r \colon E \to P^2$ given by $e \mapsto (e^+, e^-)$ with $e^+ \preceq e^-$,
we define a set ${\sup}_r(I)$ by
\begin{align}
    {\sup}_r(I) \coloneqq \{ e \in E \mid e^+ \in I \cup \{ \bot \} \not\ni e^- \}
\end{align}
for any ideal $I \in \cI(P^*)$ of $P^*$.
The second property~\eqref{cond:R} states that every instance $\mathbf{I}$ of {\sc Prob} admits a map
$r \colon E \to P^2$
such that ${\sup}_r$ forms a surjection from the family $\cI(P^*)$ of ideals of $P^*$ to the family $\cS(\mathbf{I})$ of solutions of $\mathbf{I}$.
\begin{enumerate}[label=(\condR),ref=\condR]
\item \label{cond:R}
For any instance $\mathbf{I}$ of {\sc Prob}, there exist a finite poset
$(P, \preceq)$
with different minimum element $\bot$ and maximum element $\top$ and a map $r \colon E \to P^2$ given by $e \mapsto (e^+, e^-)$ with $e^+ \preceq e
^-$ such that
\begin{align}
    \cS(\mathbf{I}) = \{ {\sup}_r(I) \mid I \in \cI(P^*) \}.
\end{align}
\end{enumerate}
We refer to a map $r$ in the property~\eqref{cond:R} as a \emph{reduction map}.

Our framework exploits network flow techniques to efficiently compute diverse solutions.
To utilize them, we will construct a directed acyclic graph (DAG) to represent a poset.
Here, we say that a DAG $G$ \emph{represents} a poset $(P,\preceq)$ with different minimum element $\bot$ and maximum element $\top$
if the vertex set of $G$ is $P$ with $\top$ and $\bot$ being a source and sink, respectively,
and
for any $u,v \in P^*$, there is a $u$-$v$ path in $G$ if and only if $u \succeq v$.

Our main result is stated as follows.
We remark that we focus on deterministic algorithms in this paper.

\begin{theorem}\label{thm:general}
    Suppose that a combinatorial problem {\sc Prob} has the properties~\eqref{cond:S} and~\eqref{cond:R}.
    Then,
    for any instance
    $\mathbf{I}$ of {\sc Prob},
    the problems
    \SumkDiverse and \CovkDiverse of {\sc Prob} can be solved in
    $O(T_P(\mathbf{I}) + T_r(\mathbf{I}) + (|A| + k|E|)^{1+o(1)})$ time
    and in
    $O(T_P(\mathbf{I}) + T_r(\mathbf{I}) + (|A| + |E|)^{1+o(1)} \log^2 k + kq)$ time,
respectively.
Here,
$T_P(\mathbf{I})$ is the time required to construct a DAG representing the poset $P$ in~\eqref{cond:R},
$T_r(\mathbf{I})$ is the time required to construct a reduction map $r$ in~\eqref{cond:R},
and
$A$ is the set of arcs of a DAG $G$ that represents $P$.
\end{theorem}

We obtain \Cref{thm:general} via the reduction of \SumkDiverse/\CovkDiverse of {\sc Prob} to a classical network flow problem, called the \emph{minimum cost flow problem}.
This enables us to utilize the state-of-the-art algorithm~\cite{brand2023} for the minimum cost flow problem in solving \SumkDiverse/\CovkDiverse of {\sc Prob}.
To this end, we introduce an intermediate problem called the \emph{minimum $k$-potential problem}, and reduce the $k$-diverse problems to the minimum $k$-potential problem.
Then we further reduce the minimum $k$-potential problem to the minimum cost flow by utilizing the idea of Ahuja, Hochbaum, and Orlin~\cite{ahuja2003}, who dealt with a more general problem.
Our novelty is to identify the properties~\eqref{cond:S} and~\eqref{cond:R} as a sufficient condition to reduce the $k$-diverse problem to the minimum $k$-potential problem.

We also provide the reduction of
\SumkDiverse/\CovkDiverse of {\sc Prob} to another classical network flow problem, called the \emph{maximum $s$-$t$ flow problem} (or the \emph{minimum $s$-$t$ cut problem}).
The construction is simple: Make $k$ copies of the digraph representing a poset, and reverse and add arcs so that any minimum $s$-$t$ cut corresponds to an optimal set of diverse solutions.
The drawback is that, even when we use the state-of-the-art algorithm~\cite{brand2023} for the maximum $s$-$t$ flow problem, the running-time of this reduction is slightly worse than that in~\cref{thm:general}.
However,
the actual running-time depends on the implementation of network flow algorithms.
Since the maximum $s$-$t$ flow problem is simpler than the minimum cost flow, it is possible that this reduction performs faster in practice.

As applications of our framework, we demonstrate that two classical combinatorial problems, \Mincut and \StableMatching, have properties~\eqref{cond:S} and~\eqref{cond:R}.
Here, \Mincut is the problem of finding an $s$-$t$ cut of a given digraph $G$ having the minimum size,
and \StableMatching is the problem of finding a stable matching of given preference lists;
formal definitions are given in \Cref{subsec:k-diverse-mincut,subsec:k-diverse-stable-matching}.
By applying~\cref{thm:general} to \SumkDiverse/\CovkDiverse of the two problems, we obtain the following results.

\begin{theorem}\label{thm:mincut-stablematching}
    \begin{enumerate}[{label={\textup{(\arabic*)}}}]
        \item The problems \SumkDiverse and \CovkDiverse of \Mincut can be solved
        in $O(n + (km)^{1 + o(1)})$ time
        and
        in $O(n + m^{1 + o(1)} \log^2 k + kq)$ time,
        respectively,
        where $n$ denotes the number of vertices of the input digraph $G$,
        $m$ the number of arcs of $G$,
        and $q$ the size of a minimum $s$-$t$ cut of $G$.
        \item The problems \SumkDiverse and \CovkDiverse of \StableMatching can be solved
        in $O((kn^2)^{1 + o(1)})$ time
        and
        in $O(n^{2 + o(1)}\log^2 k + kn)$ time,
        respectively,
        where $n$ denotes the size of the ground set $U$ (or $V$) of the input instance $(U, V; (\leq_u)_{u \in U}, (\leq_v)_{v \in V})$.
    \end{enumerate}
\end{theorem}

The polynomial-time solvability of \SumkDiverse/\CovkDiverse of \Mincut has already been shown by De Berg, Mart\'{i}nez, and Spieksma~\cite{deBerg2023}.
Very recently, independently of our work\footnotemark, in~\cite{deBerg2025}, the same authors develop a framework for solving \SumkDiverse/\CovkDiverse of a combinatorial problem having a certain property (we mention it as the property~\eqref{cond:T} in \Cref{subsec:k-diverse-totalorder}).
\footnotetext{The present paper is based on our graduation thesis~\cite{morihira2024} and master thesis~\cite{matsuda2025}.}
They show by using their framework that \SumkDiverse/\CovkDiverse of \StableMatching can be solved in polynomial time.
Both of their tractability results in~\cite{deBerg2023,deBerg2025} are based on the polynomial-time solvability of the \emph{submodular function minimization (SFM)},
which is a general and abstract problem appearing in the fields of combinatorial optimization and theoretical computer science~\cite{Bach2013-ti,Fujishige2005-sp}.
Even if we use the state-of-the-art algorithm for SFM given in~\cite{Jiang2021},
the running-time of their algorithms is (polynomial but) not very fast; it takes $O(k^5n^5)$ time for \SumkDiverse/\CovkDiverse of \Mincut, and $O(k^5n^9)$ time for \SumkDiverse/\CovkDiverse of \StableMatching.
As in \Cref{thm:mincut-stablematching}, our proposed algorithms for \SumkDiverse/\CovkDiverse of \Mincut/\StableMatching are,
thanks to exploiting network flow structure,
much faster than the previous ones.
Moreover, we show that the framework of De Berg, Mart\'{i}nez, and Spieksma for \SumkDiverse and \CovkDiverse can be captured by our framework,
which allows us to improve the time complexity.

\paragraph{Related work.}
The problem of finding diverse solutions dates back at least 1970's~\cite{glover1977}.
Nowadays, there exists a vast body of literature on the problem of finding diverse solutions, and we mention just a few papers here.
Finding diverse solutions is generally harder than finding a single one.
In early work, finding diverse points with the maximum norm distances was shown to be NP-hard~\cite{kuo1993}.
The $k$-diverse problems with respect to $\dsum$ of some polynomial-time solvable problems, such as the maximum matching problem~\cite{fomin2024a} and the minimum cut problem~\cite{Hanaka2023}, are known to be NP-hard.
Thus, in recent years, there has been a lot of papers on fixed-parameter tractable (FPT) algorithms for NP-hard diverse problems~\cite{baste2022, baste2019,eiben2024, fomin2024a, fomin2024b, hanaka2021}.
Meanwhile, polynomial-time algorithms are known for some other $k$-diverse problems with respect to the diversity measure $\dsum$ or $\dcov$,
including
those of the spanning-tree problem~\cite{hanaka2021} and the shortest $s$-$t$ path problem~\cite{hanaka2022} as well as \Mincut~\cite{deBerg2023} and \StableMatching~\cite{deBerg2025}, which have been already mentioned.

\paragraph{Organization.}
In \Cref{sec:preliminaries}, we introduce basic definitions and notation.
In \Cref{sec:reduction},
we reduce \SumkDiverse/\CovkDiverse of a combinatorial problem having the properties~\eqref{cond:S} and~\eqref{cond:R} to the minimum $k$-potential problem.
In \Cref{sec:reduction-to-mcf},
we solve the minimum $k$-potential problem by using network flow algorithms; the minimum cost flow (\Cref{subsec:mincostflow}) and the maximum flow (\Cref{subsec:maxflow}).
In \Cref{sec:applications}, we describe application of our framework.
We first present a common strategy for constructing a reduction map (\Cref{subsec:strategy}), and then demonstrate the construction for the $k$-diverse problems of \Mincut (\Cref{subsec:k-diverse-mincut}) and \StableMatching (\Cref{subsec:k-diverse-stable-matching}).
We also show in \Cref{subsec:k-diverse-totalorder} that our framework can capture the framework of De Berg, Mart\'{i}nez, and Spieksma~\cite{deBerg2025} in the case of \SumkDiverse/\CovkDiverse.

\section{Preliminaries}\label{sec:preliminaries}
Let $\Z$, $\Z_+$, $\R$, and $\R_+$ denote the set of integers, nonnegative integers,
real numbers,
and nonnegative real numbers, respectively.  
For a nonnegative integer $k$, let $[k] \coloneqq \{1, 2, \dots, k\}$
and $[0,k] \coloneqq \{0,1,\dots, k\}$.
The symmetric difference $(S \setminus T) \cup (T \setminus S)$ of sets $S$ and $T$ is denoted by $S \triangle T$. 
For a finite set $E$ and a nonnegative integer $q \in \Z_+$,
let $\binom{E}{q}$ denote the family of subsets of $E$ with size $q$, i.e.,
$\binom{E}{q} \coloneqq \{ S \subseteq E \mid |S| = q \}$.
For any finite set $E$, element $e \in E$, and $k$-tuple $\mathbf{S} = (S_1, S_2, \dots, S_k)$ of subsets $S_1, S_2, \dots, S_k$ of $E$,
the multiplicity of $e$ with respect to $\mathbf{S}$, denoted by $\mu_e(\mathbf{S})$, is the number of subsets $S_i$ that contains $e$,
i.e.,
$\mu_e(\mathbf{S}) \coloneqq | \{ i \in [k] \mid e \in S_i \} |$.

A \emph{partially ordered set} (or \emph{poset}) is a pair $(P, \preceq)$ of a set $P$ and a binary relation $\preceq$ over $P$ satisfying, for $x,y,z \in P$, that $x \preceq x$ (reflexivity), $x \preceq y$ and $y \preceq x$ imply $x = y$ (antisymmetry),
and $x \preceq y$ and $y \preceq z$ imply $x \preceq z$ (transitivity).
By $x \prec y$ we mean $x \preceq y$ and $x \neq y$.
We refer to the binary relation $\preceq$ as the \emph{partial order}.
If a partial order $\preceq$ on $P$ is clear from the context, we denote by $P$ a poset and its underlying set interchangeably.
In this paper, we only consider a finite poset, i.e.,
a poset whose underlying set is finite.
Hence,
by a poset we mean a finite poset.
A partial order $\preceq$ on $P$ is called a \emph{total order}
if $x \preceq y$ or $y \preceq x$ holds for any $x, y \in P$.
For
a poset $(P, \preceq)$,
a subset $I \subseteq P$ is called an \emph{ideal}
if $I$ is closed under $\preceq$,
i.e.,
for any $v \in P$ and $u \preceq v$,
we have $u \in P$.
Let $\cI(P)$ denote the set of all ideals of $P$.
If $P$ has the minimum element $\bot$ and maximum element $\top$,
then we denote by $P^*$ the poset obtained from $P$ by removing $\bot$ and $\top$.

Let $G = (V, A)$ be a digraph.
For a vertex subset $X \subseteq V$,
let $\Delta^+_G(X)$ (resp.~$\Delta^-_G(X)$) denote the set of outgoing arcs from (resp.~incoming arcs to) $X$.
If no confusion arises, we omit the subscript $G$ from $\Delta^+_G(X)$ and $\Delta^-_G(X)$.
If $X$ consists of a single vertex $v$, i.e., $X = \{v\}$,
we simply write $\Delta^+(v)$ and $\Delta^-(v)$ instead of $\Delta^+(\{v\})$ and $\Delta^-(\{v\})$,
respectively.
A function $f \colon A \to \Z_+$ from the arc set $A$ to the nonnegative integers is called a \emph{flow} of $G$.
In the case where the digraph $G$ has an arc capacity $c \colon A \to \Z_+$,
a flow $f \colon A \to \Z_+$ is said to be \emph{feasible} (with respect to $c$)
if $f(a) \leq c(a)$ for all $a \in A$.
For a feasible flow $f$ with respect to an arc capacity $c$,
its \emph{residual graph}, denoted by $G_f$, is the digraph whose vertex set is $V$
and arc set is $\{ a \mid a \in A,\ f(a) < c(a) \} \cup \{ a \mid \bar{a} \in A,\ 0 < f(\bar{a}) \}$,
where $\bar{a}$ denotes the reverse arc $(v, u)$ of $a = (u,v)$.
This plays an important role in algorithms for network flow problems (and our algorithms).
For a flow $f \colon A \to \Z_+$,
its \emph{boundary} $\partial f \colon V \to \Z$ is defined by
$\partial f(v) \coloneqq \sum_{a \in \Delta^+(v)} f(a) - \sum_{a \in \Delta^-(v)} f(a)$
for each $v \in V$.

In the following (except for \Cref{subsec:k-diverse-totalorder}), we assume that {\sc Prob} is a combinatorial problem having the properties~\eqref{cond:S} and~\eqref{cond:R}.
For an instance $\mathbf{I}$ of {\sc Prob},
we denote by
$\cS(\mathbf{I})$ the family of solutions of $\mathbf{I}$, 
$T_P(\mathbf{I})$ the time required to construct a DAG representing the poset $P$ in~\eqref{cond:R}, and 
$T_r(\mathbf{I})$ the time required to construct a reduction map $r$ in~\eqref{cond:R}.

\section{Reduction to the minimum \texorpdfstring{$k$}{k}-potential problem}
\label{sec:reduction}
In this section, we provide a reduction from the $k$-diverse problem of
{\sc Prob}
to the \emph{minimum $k$-potential problem},
which we introduce later.

In \SumkDiverse and \CovkDiverse of {\sc Prob},
we can regard the diversity measures $\dsum$ and $\dcov$ as the functions over $(\cS(\mathbf{I}))^k$ for each instance $\mathbf{I}$ of {\sc Prob}.
Since $\cS(\mathbf{I}) \subseteq \binom{E}{q}$ holds by the property~\eqref{cond:S},
we have $\sum_{e \in E} \mu_e(\mathbf{S}) = \sum_{i = 1}^k|S_i| = kq$ for any $\mathbf{S} = (S_1, S_2, \dots, S_k) \in (\cS(\mathbf{I}))^k$,
which is a constant.
Hence,
the functions $\dsum$ and $\dcov$ are representable as
\begin{align}
    \dsum(\mathbf{S})
    &=\sum_{e\in E} \mu_e(\mathbf{S})(k-\mu_e(\mathbf{S}))
  =\mathrm{Const.}-\sum_{e\in E}\mu_e(\mathbf{S})^2,\\
\dcov(\mathbf{S}) &= \sum_{e \in E} \min \{ 1, \mu_e(\mathbf{S}) \}
    = \mathrm{Const.} -\sum_{e \in E} \max\{0, \mu_e(\mathbf{S})-1\}
\end{align}
for each $\mathbf{S} \in (\cS(\mathbf{I}))^k$.
Thus, the problems \SumkDiverse and \CovkDiverse of {\sc Prob},
namely,
the problems of maximizing the functions $\dsum$ and $\dcov$ over $(\cS(\mathbf{I}))^k$,
are equivalent to those of minimizing 
\begin{align}
\dsum^*(\mathbf{S}) \coloneqq \sum_{e \in E}\mu_e(\mathbf{S})^2, \qquad
\dcov^*(\mathbf{S}) \coloneqq \sum_{e\in E} \max\{0, \mu_e(\mathbf{S})-1\},
\end{align}
respectively.

By using the concept of discrete convex functions,
we can uniformly handle these functions $\dsum^*$ and $\dcov^*$ as follows.
A function $\varphi \colon \Z \to \Z$
is said to be \emph{discrete convex}~\cite[Chapter~3.4]{Murota2003-yb}
if
$\varphi(x - 1) + \varphi(x + 1) \geq 2 \varphi(x)$ for all $x \in \Z$,
and said to be \emph{non-decreasing on $\Z_+$}
if $\varphi(x) \leq \varphi(x+1)$ for all $x \in \Z_+$.
For a discrete convex function $\varphi$ with $\varphi(0) = 0$ that is non-decreasing on $\Z_+$,
we define
\begin{align}
    d_\varphi^*(\mathbf{S}) \coloneqq \sum_{e \in E} \varphi(\mu_e(\mathbf{S}))
\end{align}
for $\mathbf{S} \in (\cS(\mathbf{I}))^k$.
Since the functions $x \mapsto x^2$ and $x \mapsto \max \{ 0, x -1 \}$ are discrete convex functions that is non-decreasing on $\Z_+$ and satisfy $0 \mapsto 0$,
both of $\dsum^*$ and $\dcov^*$ admit such representations.

Our framework can be applied to the $k$-diverse problem with respect to the
measure of diversity of the form
\begin{align}
    d_\varphi(\mathbf{S}) \coloneqq \mathrm{Const.} - d_\varphi^*(\mathbf{S})
\end{align}
for $\mathbf{S} \in (\cS(\mathbf{I}))^k$,
where $\varphi$ is a discrete convex function with $\varphi(0) = 0$ that is non-decreasing on $\Z_+$.
In the following (except for \Cref{subsec:maxflow}),
we consider the $k$-diverse problem with respect to $d_\varphi$ of {\sc Prob},
or equivalently,
the problem of minimizing $d_\varphi^*$ over $(\cS(\mathbf{I}))^k$
for a discrete convex function $\varphi$.

We then introduce the \emph{minimum $k$-potential problem}, to which we reduce the problem of minimizing $d_\varphi^*$ later.
Let $G = (V, A)$ be a DAG having unique source vertex ${\top}$ and sink vertex ${\bot}$ with $\top \neq \bot$.
We refer to an assignment $p\colon V \to \mathbb{Z}_+$ of integers to vertices as a \emph{$k$-potential} if $p$ satisfies the following conditions:
\begin{enumerate}[label=(P\arabic*), ref=P\arabic*]
  \item \label{cond:P1}
  $p(\bot)=k$ and $p(\top)=0$.
  \item \label{cond:P2}
  $0\leq p(v)\leq k$ for each $v \in V$.
  \item \label{cond:P3}
  $p$ is monotone non-increasing with respect to $A$, i.e., $p(u)\leq p(v)$ for each $(u,v)\in A$.
\end{enumerate}
In the \emph{minimum $k$-potential problem},
we are given a DAG $G = (V, A)$ having unique source ${\top}$ and sink ${\bot}$, an arc weight $w \colon A \to \Z_+$, and a discrete convex function $\varphi \colon \Z \to \Z$ with $\varphi(0)=0$ that is non-decreasing on $\Z_+$.
The problem asks to find a $k$-potential $p$ of $G$ that minimizes
\begin{align}\label{eq:k-optimal-potential}
    H(p) \coloneqq \sum_{a = (u,v) \in A} w(a) \varphi(p(v)-p(u)).
\end{align}
We note that the function $H$ does not change even if we remove vertices $v \in V \setminus \{ \bot, \top \}$ with $\Delta^+(v) = \Delta^-(v) = \emptyset$,
called \emph{isolated vertices}, from $G$.

Finally, we reduce,
for an instance $\mathbf{I}$ of {\sc Prob},
the problem of minimizing $d_\varphi^*$ over $(\cS(\mathbf{I}))^k$ to the minimum $k$-potential problem by utilizing the property~\eqref{cond:R} as follows.
Let
$(P, \preceq)$
be a poset having minimum element $\bot$ and maximum element $\top$,
and let $r \colon E \to P^2$ be a reduction map as in the property~\eqref{cond:R}.
Then, we construct a DAG $G_\mathbf{I}$ whose vertex set is $V = P$
and whose arc set is $A = A_P \cup A_E$,
where $A_P$ is an arc set such that a DAG $(P, A_P)$ represents the poset $(P, \preceq)$,
and
$A_E \coloneqq \{ (e^-, e^+) \mid e \in E \}$.
We can easily see that the resulting $G_{\mathbf{I}}$ is still a DAG that represents $P$
and has unique source $\top$ and sink $\bot$;
each arc $(e^-, e^+) \in A_E$ is compatible with the partial order $\preceq$ of $P$, since $e^- \succeq e^+$.
An arc weight $w_\mathbf{I} \colon A \to \Z_+$ is defined by
$w_\mathbf{I}(a) \coloneqq |\{ e \in E \mid a = (e^-, e^+) \}|$
for $a \in A$.
Note that $w_\mathbf{I}(a) = 0$ for each $a \in A \setminus A_E$.
We set $\varphi$, which satisfies the non-decreasing property on $\Z_+$ and $\varphi(0) = 0$,
as the input discrete convex function of the minimum $k$-potential problem.
Then, the triple $(G_\mathbf{I}, w_\mathbf{I}, \varphi)$ is an instance of the minimum $k$-potential problem;
its construction time is $|E|$.

Intuitively, any $k$-potential represents the direct sum of $k$ ideals $I_1, I_2, \dots, I_k \in \cI(P^*)$ of $P^*$ as a multiset.
Conversely, for any $k$ ideals in $\cI(P^*)$, there exists a $k$-potential that represents their direct sum as a multiset.
Furthermore, by the property~\eqref{cond:R},
each ideal of $P^*$ corresponds to a solution of $\mathbf{I}$ via ${\sup}_r$.
The following lemma verifies this intuition.

\begin{lemma}\label{lem:d=H}
    For each $k$-tuple $\mathbf{S} \in (\cS(\mathbf{I}))^k$,
    there is a $k$-potential $p_{\mathbf{S}}$ of $G_\mathbf{I}$
    such that
    $H(p_{\mathbf{S}}) = d_\varphi^*(\mathbf{S})$.
    Conversely,
    for each $k$-potential $p$ of $G_\mathbf{I}$, there is a $k$-tuple $\mathbf{S}_p \in (\cS(\mathbf{I}))^k$ of solutions of $\mathbf{I}$ such that
    $H(p) = d_\varphi^*(\mathbf{S}_p)$.
    In particular,
    we can construct $\mathbf{S}_p$ from $p$
    in $O(|E| + kq)$ time.
\end{lemma}
\begin{proof}
    Take any $k$-tuple $\mathbf{S} = (S_1, S_2, \dots, S_k) \in (\cS(\mathbf{I}))^k$.
    By the property~\eqref{cond:R},
    i.e., the surjectivity of ${\sup}_r \colon \cI(P^*) \to \cS(\mathbf{I})$,
    for each $i \in [k]$
    there is an ideal $I_i \in\cI(P^*)$ of $P^*$
    such that ${\sup}_r(I_i) = S_i$.
    We define $p_{\mathbf{S}} \colon P \to \Z_+$ by
    $p_{\mathbf{S}}(v) \coloneqq |\{ i \in [k] \mid v \in I_i \cup \{ \bot \} \}|$
    for $v \in P$.
    This $p_{\mathbf{S}}$ is a $k$-potential.
    Indeed, \eqref{cond:P1} and~\eqref{cond:P2} clearly hold.
    Since each $I_i \cup \{ \bot \}$ is an ideal of $P$ and $(u, v) \in A$ implies $u \succeq v$,
    if $(u, v) \in A$ and $u \in I_i \cup \{ \bot \}$,
    then $v \in I_i \cup \{ \bot \}$.
    Hence, we obtain $p(u) \leq p(v)$ for each $(u,v) \in A$; \eqref{cond:P3} holds.

    Since
    $\mu_e(\mathbf{S}) = |\{ i \in [k] \mid e^+ \in I_i \cup \{\bot\} \not\ni e^- \}|
    = |\{ i \in [k] \mid e^+ \in I_i \cup \{\bot\} \}| - |\{ i \in [k] \mid e^- \in I_i \cup \{\bot\} \}|
    = p_{\mathbf{S}}(e^+) - p_{\mathbf{S}}(e^-)$,
    we obtain
    \begin{align}
        d_\varphi^*(\mathbf{S}) &= \sum_{e \in E} \varphi(\mu_e(\mathbf{S}))\\
        &= \sum_{e \in E} \varphi(p_{\mathbf{S}}(e^+) - p_{\mathbf{S}}(e^-))\\
        &= \sum_{a = (u,v) \in A} w_\mathbf{I}(a) \varphi(p_{\mathbf{S}}(v)-p_{\mathbf{S}}(u))\\
        &= H(p_{\mathbf{S}}),\label{eq:d=H}
    \end{align}
    where the third equality follows from the definition of the arc weight $w_\mathbf{I}(a) \coloneqq |\{ e \in E \mid a = (e^-, e^+) \}|$.
    
    Conversely, take any $k$-potential $p$ of $G$.
    We define $k$ subsets $I_1, I_2, \dots, I_k$ of $P$ by
    $I_i \coloneqq \{ v \in P \mid p(v) \geq i \}$
    for each $i \in [k]$.
    Then, each $I_i$ is an ideal of $P$.
    Indeed, take any $v \in I_i$ (i.e., $p(v) \geq i$) and $u \preceq v$.
    By the construction of $G$,
    there is a $v$-$u$ path in $G$,
    which implies that $p(u) \geq p(v)$ by~\eqref{cond:P3}.
    Hence, we obtain $p(u) (\geq p(v)) \geq i$,
    i.e., $u \in I_i$.
    We can easily observe that
    $p(v) = |\{ i \in [k] \mid v \in I_i \}|$
    for each $v \in P$.
    Since $p(\bot) = k$ and $p(\top) = 0$ by~\eqref{cond:P1},
    we obtain $\bot \in I_i \not\ni \top$
    for each $i \in [k]$.

    Let $\mathbf{S}_p \coloneqq ({\sup}_r(I_1 \setminus \{ \bot \}), {\sup}_r(I_2 \setminus \{ \bot \}), \dots, {\sup}_r(I_k \setminus \{ \bot \})) \in (\cS(\mathbf{I}))^k$.
    Since, for each $e \in E$, we have
    $\mu_e(\mathbf{S}_p) = |\{ i \in [k] \mid e^+ \in I_i \not\ni e^- \}|
    = |\{ i \in [k] \mid e^+ \in I_i \}| - |\{ i \in [k] \mid e^- \in I_i \}|
    = p(e^+) - p(e^-)$,
    we obtain $d_\varphi^*(\mathbf{S}_p) = H(p)$
    from the reverse argument of~\eqref{eq:d=H} by replacing $\mathbf{S}$ and $p_{\mathbf{S}}$ with $\mathbf{S}_p$ and $p$, respectively.
    
    Since ${\sup}_r(I_i \setminus \{ \bot \}) = \{ 
e \in E \mid p(e^+) \geq i > p(e^-) \}$,
we can directly construct $\mathbf{S}_p$ from $p$ as follows.
We prepare $k$ emptysets $S_1, S_2, \dots, S_k$,
and then, for each $e \in E$,
we add $e$ to $S_i$ for all $i$ with $p(e^-) < i \leq p(e^+)$;
the resulting $k$-tuple $(S_1, S_2, \dots, S_k)$ coincides with $\mathbf{S}_p$.
The above construction of the $k$-tuple $(S_1, S_2, \dots, S_k)$ takes $O(|E| + kq)$ time,
since $\sum_{i = 1}^k|S_i| = kq$.
\end{proof}

\Cref{lem:d=H} immediately implies that
we can construct in $O(|E| + kq)$ time a minimizer $\mathbf{S}_p$ of $d_\varphi^*$ over $(\cS(\mathbf{I}))^k$
from a minimum $k$-potential $p$ for the instance $(G_\mathbf{I}, w_\mathbf{I}, \varphi)$ of the minimum $k$-potential problem.
Therefore, we obtain the following.
\begin{theorem}\label{thm:sec2:k-diverse}
We can solve the $k$-diverse problem with respect to $d_\varphi$ of {\sc Prob}
    in $O(T_P(\mathbf{I}) + T_r(\mathbf{I}) + T_\textup{mp}(\mathbf{I}) + |E| + kq)$ time,
    where $\mathbf{I}$ is a given instance of {\sc Prob}
    and $T_\mathrm{mp}(\mathbf{I})$ is
    the time of solving the instance of the minimum $k$-potential problem reduced from $\mathbf{I}$.
\end{theorem}

Since isolated vertices in $G_{\mathbf{I}}$ except for $\bot$ and $\top$ do not affect the value of function $H$, we can remove them before solving the instance of the minimum $k$-potential problem reduced from $\mathbf{I}$.
This elimination can be performed in linear time in the size of $G$, which is upper-bounded by $T_P(\mathbf{I})$.
Therefore, we may assume that $G_{\mathbf{I}}$ has no such isolated vertices, implying that $|P| = O(|A_P| + |E|)$.
This assumption will be used in estimating the time complexity of $T_\mathrm{mp}(\mathbf{I})$ in \Cref{subsec:mincostflow}.

\section{Solving the minimum \texorpdfstring{$k$}{k}-potential problem via network flow}\label{sec:reduction-to-mcf}
In this section, we describe reductions of the minimum $k$-potential problem to network flow problems, which enables us to use the state-of-the-art network flow algorithms for solving the minimum $k$-potential problem, or the original $k$-diverse problem.
We provide two reductions: One is to the minimum cost flow problem given in \Cref{subsec:mincostflow};
the other to the maximum $s$-$t$ flow problem (or minimum $s$-$t$ cut) given in \Cref{subsec:maxflow}.
We need to remark here that the latter reduction is valid only for the case where the input discrete convex function $\varphi$ of the minimum $k$-potential problem is either $x \mapsto x^2$ or $x \mapsto \max\{ 0, x-1 \}$,
which correspond to \SumkDiverse or \CovkDiverse, respectively.

\subsection{Via minimum cost flow}\label{subsec:mincostflow}
This subsection is devoted to the reduction of the minimum $k$-potential problem to the minimum cost flow problem.
Let us first recall the \emph{minimum cost flow problem} (see e.g.,~\cite[Chapter~9]{Korte2018-ii} and~\cite[Chapter~12]{schrijver2003} for details).
In this problem,
we are given a digraph $G = (V, A)$, arc cost $\gamma \colon A \to \Z$,
arc capacity $c \colon A \to \Z_+$,
and vertex demand $d \colon V \to \Z$,
and
asked to find a feasible flow $f \colon A \to \Z_+$
that minimizes $\sum_{a \in A} \gamma(a) f(a)$ subject to
$\partial f(v) = d(v)$ for all $v \in V$.
An optimal solution of the minimum cost flow problem is called a \emph{minimum cost flow}.
We denote by $T_\mathrm{mcf}(n, m, \Gamma, C, D)$ the time required to solve the minimum cost flow problem for a network of $n$ vertices and $m$ arcs with cost at most $\Gamma$ in absolute values, capacity at most $C$, and a demand vector with values at most $D$ in absolute values.
Using the state-of-the-art algorithm for the minimum cost flow problem given in~\cite{brand2023},
we have $T_\mathrm{mcf}(n, m, \Gamma,C, D) = O(m^{1+o(1)} \log (\max\{ C, D \}) \log \Gamma)$.

Our reduction follows the work
by Ahuja, Hochbaum, and Orlin~\cite{ahuja2003}, who dealt with a more general problem called the \emph{convex cost integer dual network flow problem}.
They showed that
the Lagrangian dual
of their problem is reduced to the minimum cost flow problem.
For the complexity analysis, we provide the actual reduction in our case and the construction of a minimum $k$-potential from a minimum cost flow.
Here, we introduce some terminology.
For a discrete convex function $\psi \colon \Z \to \Z$,
an integer $x \in \Z$ is called a \emph{breakpoint} of $\psi$
if $\psi(x+1) +\psi(x-1) > 2\psi(x)$, i.e.,
the \emph{left slope} $\psi(x) - \psi(x-1)$ and \emph{right slope} $\psi(x+1) - \psi(x)$ of $\psi$ at $x$ are different.
We can observe that,
for each $x \in \Z$,
the left slope of $\psi$ at $x$ is at most the right slope of $\psi$ at $x$.
Let $B(\psi)$ denote the set of breakpoints of $\psi$.

Let $(G = (V, A), w \colon A \to \Z_+, \varphi \colon \Z \to \Z)$ be an instance of the minimum $k$-potential problem.
We define $B_k(\varphi) \coloneqq \left(B(\varphi) \cap [0,k]\right) \cup \{ 0,k \}$
and suppose that $B_k(\varphi) = \{ b_0, b_1, \dots, b_z \}$ with $(0 =) b_0 < b_1 < \cdots < b_z (= k)$.
Let $s_i^-$ (resp.\ $s_i^+$) denote the left (resp.\ right) slope of $\varphi$ at $b_i \in B(\varphi) \cap [k-1]$;
note that $s_i^+ = s_{i+1}^- < s_{i+1}^+$.
We set $M \in \Z_+$ as a sufficiently large integer satisfying $M > H(p)$ for any $k$-potential $p$ of $G$,
e.g., $M = \sum_{a\in A} w(a) \varphi(k) + 1$.

We construct an instance of the minimum cost flow problem as follows.
The vertex set of the input digraph $\bar{G}$ is $\bar{V} \coloneqq V \cup \{0\}$.
We set the arc set $\bar{A}$ of $\bar{G}$, arc cost $\bar{\gamma} \colon \bar{A} \to \Z$, and capacity $\bar{c} \colon \bar{A} \to \Z_+$ by creating
\begin{itemize}
    \item $|B_k(\varphi)| = z + 1$ copies of each $a \in A$ satisfying $w(a) > 0$ with costs $b_0, b_1, \dots, b_{z-1}, b_z$ and capacities
    $w(a)s_1^-+M, w(a)(s_1^+ - s_1^-), w(a)(s_2^+ - s_2^-), \dots, w(a)(s_{z-1}^+ - s_{z-1}^-),M-w(a)s_{z-1}^+$, respectively,
    \item two copies of each $a \in A$ satisfying $w(a) = 0$ or each $a = (0,u)$ for $u \in V \setminus \{ \bot, \top\}$ with costs $0, k$ and capacities $M, M$, respectively,
    \item an arc $(0, \bot)$ with cost $k$ and capacity $2M$, and
    \item an arc $(0, \top)$ with cost $0$ and capacity $2M$.
\end{itemize}
The vertex demand $\bar{d} \colon \bar{V} \to \Z$ is set as
$\bar{d}(v) \coloneqq M(|\Delta^+_{G'}(v)| - |\Delta^-_{G'}(v)|)$ for each $v\in \bar{V}$,
where $G' \coloneqq (\bar{V}, A')$ and $A' \coloneqq A \cup \{ (0,v) \mid v \in V \}$.

Suppose that we obtain a minimum cost flow $f^*$ of the resulting instance $(\bar{G}, \bar{\gamma}, \bar{c}, \bar{d})$.
Then we construct the residual graph $\bar{G}_{f^*}$ of $\bar{G}$ with respect to $f^*$,
and set the arc length $\ell$ of $\bar{G}_{f^*}$
as $\ell(a) \coloneqq \gamma(a)$ if $a \in \bar{A}$, and $\ell(a) \coloneqq -\gamma(a)$ if $\bar{a} \in \bar{A}$.
Let us define $\bar{p}^* \colon \bar{V} \to \Z$ as a feasible potential with $\bar{p}^*(0) = 0$ in $\bar{G}_{f^*}$ with respect to arc length $\ell$, i.e.,
an assignment $\bar{p}^* \colon \bar{V} \to \Z$ satisfying
$\bar{p}^*(0) = 0$ and
$\ell(a) \geq \bar{p}^*(v) - \bar{p}^*(u)$
for each arc $a = (u,v) \in \bar{A}$.
The following lemma justifies our reduction.

\begin{lemma}\label{lem:min-cost-potential}
    The restriction $p^* \colon V \to \Z$ of $\bar{p}^*$ to $V$ forms
    a minimum $k$-potential of $(G, w, \varphi)$.
\end{lemma}

The validity of \Cref{lem:min-cost-potential} follows from exactly the same argument given in~\cite{ahuja2003}.
For the sake of completeness,
we provide a direct proof of \Cref{lem:min-cost-potential}.
\begin{proof}
    Taking the linear programming (LP) dual of the negative of (the LP relaxation of) the minimum cost flow problem for $(\bar{G}, \bar{\gamma}, \bar{c}, \bar{d})$ and adding a constant $kM$,
we obtain the following (see also \cite[Chapter~9.2]{Korte2018-ii}):
\begin{align}\label{prob:min-cost-flow-dual}
\begin{array}{ll}
\text{Minimize} & \displaystyle \sum_{a' = (u,v) \in A'} M(\bar{p}(u) - \bar{p}(v)) + \sum_{a \in \bar{A}} \bar{c}(a) \bar{z}(a) + kM \\
\text{subject to} & \bar{z}(a) \geq \max\{\bar{p}(v) - \bar{p}(u) - \bar{\gamma}(a), 0 \}\quad (a \in \bar{A}),
\end{array}
\end{align}
where $\bar{p}$ and $\bar{z}$ are variables of the above dual LP, and $A' = A \cup \{ (0,v) \mid v \in V \}$.
Since adding all-one vector $\mathbf{1}$ to $\bar{p}$ (that produces the assignment $v \mapsto \bar{p}(v) + 1$) does not change the feasibility of $\bar{p}$ and the objective value,
we may assume that $\bar{p}(0) = 0$.
Moreover, since this dual LP problem has an integral optimal solution (see e.g.,~\cite[Theorem~12.8]{schrijver2003}),
we may assume that $\bar{p}$ and $\bar{z}$ are integral variables, i.e.,
$\bar{p} \colon \bar{V} \to \Z$ with $\bar{p}(0) = 0$ and $\bar{z} \colon \bar{A} \to \Z$.
In addition,
by the nonnegativity of $\bar{c}(a)$ for $a \in \bar{A}$,
we can replace $\bar{z}(a)$ with $\max\{\bar{p}(v) - \bar{p}(u) - \bar{\gamma}(a), 0 \}$ in the above formulation.
Let
\begin{align}
    \Phi_{a'}(\bar{p}) \coloneqq
    \begin{cases}
    \displaystyle
        M(\bar{p}(0) - \bar{p}(\bot)) + \bar{c}(a') \max\{\bar{p}(\bot) - \bar{p}(0) - \bar{\gamma}(a'), 0 \} + kM & \text{if $a' = (0, \bot)$},\\
        M(\bar{p}(0)-\bar{p}(\top)) + \bar{c}(a') \max\{\bar{p}(\top) - \bar{p}(0) - \bar{\gamma}(a'), 0 \} & \text{if $a' = (0, \top)$},\\
        \displaystyle
        M(\bar{p}(u) - \bar{p}(v)) + \sum_{\text{$a \in \bar{A}:$ copy of $a'$ }} \bar{c}(a) \max\{\bar{p}(v) - \bar{p}(u) - \bar{\gamma}(a), 0 \} & \text{otherwise}
    \end{cases}
\end{align}
for each $a' = (u,v) \in A'$.
The dual LP problem is equivalent to the minimization of
\begin{align}\label{prob:bar-H}
    \bar{H}(\bar{p}) \coloneqq \sum_{a' \in A'} \Phi_{a'}(\bar{p})
\end{align}
for $\bar{p} \colon V \to \Z$ with $\bar{p}(0) = 0$.

It is known (see e.g.,~\cite[Second Proof of Corollary~9.8]{Korte2018-ii}) that any feasible potential $\bar{p}^*$ of $\bar{G}_{f^*}$ with respect to $\ell$ is an optimal solution of the dual LP, i.e., a minimizer of $\bar{H}$.
Hence, it suffices to show that the minimization of $\bar{H}$ is equivalent to the minimum $k$-potential problem for $(G, w, \varphi)$,
or more precisely, show the following claim:
If the restriction $p \colon V \to \Z$ of $\bar{p}$ to $V$ is a $k$-potential then
$\bar{H}(\bar{p}) = H(p) \ (< M)$, and
if $p$ is not a $k$-potential then $\bar{H}(\bar{p}) \geq M$.
Here, we recall that the definition of $H$ is given in~\eqref{eq:k-optimal-potential}.

In the following, we prove the above claim.
We first consider each of the functions $\Phi_{a'}(\bar{p})$.
By the definition of $\Phi_{a'}$,
for each $a' = (0,v)$ with $v \in V$, we have
\begin{align}
    \Phi_{(0,v)}(\bar{p}) =
    \begin{cases}
        M(k - \bar{p}(\bot)) + 2M \max \{ \bar{p}(\bot) - k, 0 \} & \text{if $v = \bot$},\\
        -M\bar{p}(\top) + 2M \max \{ \bar{p}(\top), 0 \} & \text{if $v = \top$},\\
        -M\bar{p}(v) + M \max \{ \bar{p}(v), 0 \} + M \max \{ \bar{p}(v) - k, 0 \} & \text{if $v \in V \setminus \{ \bot, \top \}$}.
    \end{cases}
\end{align}
We can easily see by the above formulation that
\begin{align}
    &\Phi_{(0,\bot)}(\bar{p})
    \begin{cases}
        = 0 & \text{if $\bar{p}(\bot) = k$},\\
        \geq M & \text{if $\bar{p}(\bot) \neq k$},
    \end{cases}
    \qquad
    \Phi_{(0,\top)}(\bar{p})
    \begin{cases}
        = 0 & \text{if $\bar{p}(\top) = 0$},\\
        \geq M & \text{if $\bar{p}(\top) \neq 0$},
    \end{cases}\\
    &\Phi_{(0,v)}(\bar{p})
    \begin{cases}
        = 0 & \text{if $0 \leq \bar{p}(v) \leq k$},\\
        \geq M & \text{otherwise},
    \end{cases}
    \quad (v \in V \setminus \{ \bot, \top \}).
\end{align}
For each $a \in A = A' \setminus \{ (0,v) \mid v \in V \}$,
we obtain the following by direct calculation:
\begin{align}
     \Phi_{a}(\bar{p}) =
     \begin{cases}
     -M(\bar{p}(v) - \bar{p}(u)) & \text{if $\bar{p}(v) - \bar{p}(u) < 0$},\\
     w(a) \varphi(\bar{p}(v) - \bar{p}(u)) & \text{if $0 \leq \bar{p}(v) - \bar{p}(u) \leq k$},\\
     w(a) \varphi(k) + M(\bar{p}(v) - \bar{p}(u) - k)& \text{if $\bar{p}(v) - \bar{p}(u) > k$},
     \end{cases}
\end{align}
which implies
\begin{align}
     \Phi_{a}(\bar{p})
     \begin{cases}
     = w(a) \varphi(\bar{p}(v) - \bar{p}(u)) \geq 0 & \text{if $0 \leq \bar{p}(v) - \bar{p}(u) \leq k$},\\
     \geq M& \text{otherwise}.
     \end{cases}
\end{align}
Thus, if the restriction $p$ of $\bar{p}$ to $V$ is a $k$-potential, i.e., satisfies~\eqref{cond:P1},~\eqref{cond:P2}, and~\eqref{cond:P3},
then
$\bar{H}(\bar{p}) = \sum_{a = (u,v) \in A} w(a) \varphi( \bar{p}(v) - \bar{p}(u) ) = H(p)$;
otherwise, there is $a' \in A'$ such that $\Phi_{a'}(\bar{p}) \geq M$.
In the latter case,
since $\Phi_{a'}(\bar{p}) \geq 0$ for any $\bar{p}$ and $a'$,
we have $\bar{H}(\bar{p}) \geq M$.
\end{proof}

We then show the time complexity of solving the minimum $k$-potential problem via our reduction.
The digraph $\bar{G}$ has $|V|+1 = O(|V|)$ vertices, 
and at most $2(|V|-1 + |A_{0}|) + |B_k(\varphi)||A_+| = O(|V| + |A_0| + |B_k(\varphi)||A_+|)$ arcs.
Here, $A_0$ and $A_+$ denote the set of arcs $a$ with $w(a) = 0$ and $w(a) > 0$,
respectively.
Therefore, we can construct the minimum cost flow instance $(\bar{G}, \bar{\gamma}, \bar{c}, \bar{d})$
in $O(|V| + |A_0| + |B_k(\varphi)||A_+|)$ time.
The costs and capacities are nonnegative integers at most $k$ and $\bar{C}\coloneqq \max\{2M, \max_{a\in A}w(a)\varphi(k)+M\}$, respectively.
The absolute values of the demands are at most $M|V|$.
Hence, we can obtain a minimum cost flow $f^* \colon \bar{A} \to \Z_+$ of $(\bar{G}, \bar{\gamma}, \bar{c}, \bar{d})$
in $O\left(T_\mathrm{mcf}(|V|, |V| + |A_0| + |B_k(\varphi)||A_+|, k, \bar{C}, M|V|)\right)$ time.
Moreover,
we can construct the residual graph $\bar{G}_{f^*}$ of $\bar{G}$ with respect to $f^*$ in
$O(|\bar{A}|) = O(|V| + |A_0| + |B_k(\varphi)||A_+|)$ time.
We can also obtain a feasible potential $\bar{p}^*$ in $\bar{G}_{f^*}$ with respect to $\ell$ by computing a minimum cost flow of $\bar{G}_{f^*}$ as follows.
We add a supernode $\bar{0}$ and arcs $(\bar{0}, v)$ for all $v \in \bar{V}$ to $\bar{G}_{f^*}$ (the resulting digraph is denoted as $\bar{G}_{f^*}'$).
We extend $\ell$ as $\ell((\bar{0}, v)) \coloneqq 0$ for all newly added arcs $(\bar{0}, v)$,
and set arc capacity $c'$ and vertex demand $d'$
as $c'(a) \coloneqq |V|$ for each arc $a$,
and $d'(\bar{0}) \coloneqq |\bar{V}|$
and $d'(v) \coloneqq -1$ for $v \in \bar{V}$.
Then we compute a minimum cost flow $f'$ of $(\bar{G}_{f^*}', \ell, c', d')$
and set $\bar{p}^*(v) \coloneqq f'(v) - f'(0)$ for each $v \in \bar{V}$;
the resulting $\bar{p}^* : \bar{V} \to \Z$ (with $\bar{p}^*(0) = 0$) is a feasible potential of $\bar{G}_{f^*}$ with respect to $\ell$.
This takes $O\left(T_\mathrm{mcf}(|V|, |V| + |A_0| + |B_k(\varphi)||A_+|, k, |V|, |V|)\right)$ time,
which is upper-bounded by the time to compute $f^*$.
Hence, by \Cref{lem:min-cost-potential},
we can obtain a minimum $k$-potential $p^* \colon V \to \Z_+$,
which is the restriction of $\bar{p}^*$ to $V$ by \Cref{lem:min-cost-potential},
in $O\left(T_\mathrm{mcf}(|V|, |V| + |A_0| + |B_k(\varphi)||A_+|, k, \bar{C}, M|V|)\right)$ time.

Finally, by applying the state-of-the-art algorithm~\cite{brand2023} for the minimum cost flow problem,
we obtain the following,
which implies \Cref{thm:general}.

\begin{theorem}\label{thm:optimal-potential-time-complexity}
We can solve the $k$-diverse problem with respect to $d_\varphi$ of {\sc Prob} for an instance $\mathbf{I}$
in $O(T_P(\mathbf{I}) + T_r(\mathbf{I}) + (|A_P| + |B_k(\varphi)||E|)^{1+o(1)} \log (\varphi(k))\log k + kq)$ time.
In particular,
the problems \SumkDiverse of {\sc Prob} and \CovkDiverse of {\sc Prob} can be solved in
$O(T_P(\mathbf{I}) + T_r(\mathbf{I}) + (|A_P| + k|E|)^{1+o(1)})$ time
and in
$O(T_P(\mathbf{I}) + T_r(\mathbf{I}) + (|A_P| + |E|)^{1+o(1)} \log^2 k + kq)$ time,
respectively.
\end{theorem}
\begin{proof}
Recall the notation in \Cref{sec:reduction}.
The vertex set of the input digraph of the minimum $k$-potential problem is $P$
and arc set is $A = A_P \cup A_E$,
where $(P, A_P)$ represents the poset $(P, \preceq)$
and
$|A_E| \leq |E|$.
Furthermore,
$|P| = O(|A_P| + |E|)$
by the assumption in \Cref{sec:reduction}.
We also have $\sum_{a \in A} w(a) = |E|$,
$A_0 \subseteq A_P$,
and $A_+ = A_E$.
If we set $M \coloneqq \sum_{a\in A} w(a) \varphi(k) + 1 = |E|\varphi(k) + 1$,
then $\bar{C} = \max\{2M, \max_{a\in A}w(a)\varphi(k)+M\} = 2M \leq |P|M$.
Hence, by
\Cref{thm:sec2:k-diverse}
and
$T_\mathrm{mcf}(n, m, \Gamma,C, D) = O(m^{1+o(1)} \log (\max\{ C, D \}) \log \Gamma)$,
we obtain a bound of
$O(T_P(\mathbf{I}) + T_r(\mathbf{I}) + (|A_P| + |B_k(\varphi)||E|)^{1+o(1)} \log (\varphi(k))\log k + kq)$
time for the $k$-diverse problem with respect to $d_\varphi$.

In particular,
for $\varphi(x) = \varphi_{\textup{sum}}(x) \coloneqq x^2$ (that corresponds to \SumkDiverse) and $ \varphi(x) =\varphi_{\textup{cov}}(x) \coloneqq \max \{ 0, x-1 \}$ (that corresponds to \CovkDiverse),
we have $B_k(\varphi_{\textup{sum}}) = [0,k]$, $\varphi_{\textup{sum}}(k) = k^2$,
$B_k(\varphi_{\textup{cov}}) = \{ 0, 1, k\}$, and $\varphi_{\textup{cov}}(k) = k-1$.
Thus,
we obtain
$O(T_P(\mathbf{I}) + T_r(\mathbf{I}) + (|A_P| + k|E|)^{1+o(1)})$ time
for \SumkDiverse,
and
$O(T_P(\mathbf{I}) + T_r(\mathbf{I}) + (|A_P| + |E|)^{1+o(1)} \log^2 k + kq)$ time
for \CovkDiverse.
\end{proof}

\subsection{Via maximum \texorpdfstring{$s$-$t$}{s-t} flow}\label{subsec:maxflow}
In this subsection, we describe the reduction of the minimum $k$-potential problem to the maximum $s$-$t$ flow problem, or equivalently, the minimum $s$-$t$ cut problem, focusing on the cases where $\varphi(x)$ is either $x^2$ or $\max \{0,x-1\}$;
these correspond to \SumkDiverse and \CovkDiverse, respectively.

Here, we recall the basics on the maximum $s$-$t$ flow and minimum $s$-$t$ cut problems (see e.g.,~\cite[Chapter~10]{schrijver2003} for details).
Suppose that $G = (V, A)$ is a digraph with special vertices $s, t \in V$
and $c \colon A \to \Z_+$ is an arc capacity.
A feasible flow $f \colon A \to \Z_+$ is called an \emph{$s$-$t$ flow}
if it satisfies $\partial f(v) = 0$ for $v \in V \setminus \{s,t \}$;
we note that an $s$-$t$ flow $f$ satisfies $\partial f(s) = - \partial f(t)$.
An arc subset $C \subseteq A$ is called an \emph{$s$-$t$ cut}
if there is no $s$-$t$ path in the digraph obtained from $G$ by removing all arcs in $C$.
The \emph{maximum $s$-$t$ flow problem} asks to find an $s$-$t$ flow $f \colon A \to \Z_+$ that maximizes the value $\partial f(s) (= - \partial f(t))$,
and the \emph{minimum $s$-$t$ cut problem} asks to find an $s$-$t$ cut $C \subseteq A$ that minimizes the value $\sum_{a \in C} c(a)$.
An optimal solution of the maximum $s$-$t$ flow problem (resp.\ the minimum $s$-$t$ cut problem)
is called a \emph{maximum $s$-$t$ flow} (resp.\ a \emph{minimum $s$-$t$ cut}).
We can easily observe that every minimum $s$-$t$ cut is representable as $\Delta^+(X)$ for some vertex subset $X$ with $s \in X \not\ni t$.
It is well-known (see e.g.,~\cite[Theorem~10.3]{schrijver2003}) that the maximum $s$-$t$ flow and minimum $s$-$t$ cut problems are dual to each other,
i.e.,
the value $\partial f^*(s)$ of a maximum $s$-$t$ flow $f^*$ coincides with the value $\sum_{a \in C^*} c(a)$ of a minimum $s$-$t$ cut $C^*$.
Moreover, we can obtain a minimum $s$-$t$ cut from a maximum $s$-$t$ flow $f^*$ in linear time by using the residual graph $G_{f^*}$ (see e.g.,~\cite[Corollary~10.2a]{schrijver2003}
and \Cref{lem:mincut-characterization} below).

In our reduction,
we represent the function $H(p)$ (defined in~\eqref{eq:k-optimal-potential}) over $k$-potentials $p$ as the \emph{$s$-$t$ cut function} $X \mapsto \sum_{a \in \Delta^+(X)} c(a)$ over vertex subsets $X$ with $s \in X \not\ni t$
of some digraph,
which is so-called a \emph{network representation}
(see e.g.,~\cite{Iwamasa2018}).
This reduction enables us to efficiently solve the $k$-diverse problem using state-of-the-art algorithms for the maximum $s$-$t$ flow problem.
We denote by $T_{\textup{cut}}(n, m, C)$ the time required to solve the minimum $s$-$t$ cut problem for a network of $n$ vertices and $m$ arcs with capacity at most $C$.
Using the state-of-the-art algorithm for the maximum $s$-$t$ flow problem given in~\cite{brand2023} (and the duality between the maximum $s$-$t$ flow and minimum $s$-$t$ cut problems),
we have $T_\mathrm{cut}(n, m, C) = O(m^{1+o(1)} \log C)$.

We first note that we use $\varphi(x) = \binom{x}{2}$ (where $\binom{0}{2}=\binom{1}{2}=0$) instead of $\varphi(x) = x^2$.
Since $\dsum$ is expressed in \Cref{sec:reduction} as $\dsum(\mathbf{S}) = \mathrm{Const.} - \sum_{e \in E} \mu_e(\mathbf{S})^2$ and the sum $\sum_{e \in E} \mu_e(\mathbf{S})$ is a constant $kq$ for a given instance, we can further rewrite it as
\begin{equation}
  \dsum(\mathbf{S}) = \mathrm{Const.} - 2 \sum_{e \in E} \binom{\mu_e(\mathbf{S})}{2} 
\end{equation}
for $\mathbf{S} \in (\cS(\mathbf{I}))^k$.
Thus, this choice of $\varphi$ does not affect the optimal solution of \SumkDiverse.
Although this modification does not improve the theoretical time complexity of the algorithm,
it removes some edges from the graph representing an instance of the minimum $s$-$t$ cut problem,
which can lead to faster computations in practice.

Let us see our reduction to the minimum $s$-$t$ cut problem.
Let $(G=(V,A), w, \varphi)$ be an instance of the minimum $k$-potential problem.
The idea is to make $k$ copies of the digraph $G$ and combine them into an extended graph $\hat{G}=(\hat{V}, \hat{A})$ with a source vertex $s$ and a sink vertex $t$.
We remark that we reverse all the arcs in $G$ unlike the previous subsection.
Formally, for each $v\in V$, we denote by $v^1,v^2,\ldots,v^k$ the copies of $v$.
Define $\hat{V}\coloneqq\{ v^i \mid v\in V,\ i\in [k] \} \cup \{s,t\}$ with additional vertices $s$ and $t$.
We set $\hat{A}_1 \coloneqq \{ (v^i,v^{i+1}) \mid v \in V,\ i \in [k-1] \}$ to connect copied graphs.
Let $\hat{A}_2$ be the set of \emph{reversed} copy arcs together with arcs from copies of $\bot$ to $t$ and those from $s$ to copies of $\top$, i.e., 
$\hat{A}_2 \coloneqq \{ (v^i,u^i) \mid (u,v)\in A,\ i \in [k] \}
    \cup \{ (\bot^i,t) \mid i \in [k] \}
    \cup \{ (s,\top^i) \mid i \in [k] \}$.
Let $\hat{A}_3$ be a set of arcs representing $\varphi$.
We set $\hat{A}_3 \coloneqq \{ (u^i,v^j) \mid a=(u,v)\in A,\ w(a) > 0,\ \ 1\leq i < j \leq k \}$ when $\varphi(x)=\binom{x}{2}$, and set $\hat{A}_3=\{ (u^i,v^{i+1}) \mid a=(u,v)\in A,\ w(a) > 0,\ i \in [k-1] \}$ when $\varphi(x)=\max\{0,x-1\}$.
The arc set of the extended graph is $\hat{A} \coloneqq \hat{A}_1 \cup \hat{A}_2 \cup \hat{A}_3$.
Let $M \in \Z_+$ be a sufficiently large integer satisfying $M > H(p)$ for any $k$-potential $p$ of $G$.
We set the capacity of each arc in $\hat{A}_1\cup \hat{A}_2$ to $M$ and that of each arc $(u^i, v^j) \in \hat{A}_3$ to $w(a)$, where $a = (u,v) \in A$.

We explain the role of each arc set.
For each $0$-$1$ assignment to $\hat{V}$, we assume that we incur a cost from arc $(\hat{u},\hat{v}) \in \hat{A}$ when $\hat{u}$ is assigned $0$ and $\hat{v}$ is assigned $1$.
The arcs in $\hat{A}_1$ are set to have a sufficiently large capacity so that they ensure that 
each $0$-$1$ assignment to $\hat{V}$ with small cost follows a threshold pattern: For each $v\in V$, there exists an integer $i$ such that $v^1, v^2, \ldots, v^i$ are assigned $1$, while $v^{i+1}, v^{i+2}, \ldots, v^k$ are assigned $0$.
This construction encodes an integer assignment $p(v) \coloneqq i$ for each $v\in V$.
The arcs in $\hat{A}_2$ ensure that this integer assignment $p$ to $V$ satisfies the $k$-potential conditions.
Finally, the arcs in $\hat{A}_3$ represent a cost $w(a)\varphi(p(v) - p(u))$ for each $a=(u,v) \in A$.
Figure~\ref{fig:via-mincut} illustrates the arc sets $\hat{A}_1$ and $\hat{A}_2$ and the correspondence between a $0$-$1$ assignment to $\hat{V}$ and an integer assignment to $V$.

\begin{figure}
    \centering
    \includegraphics[scale=0.45]{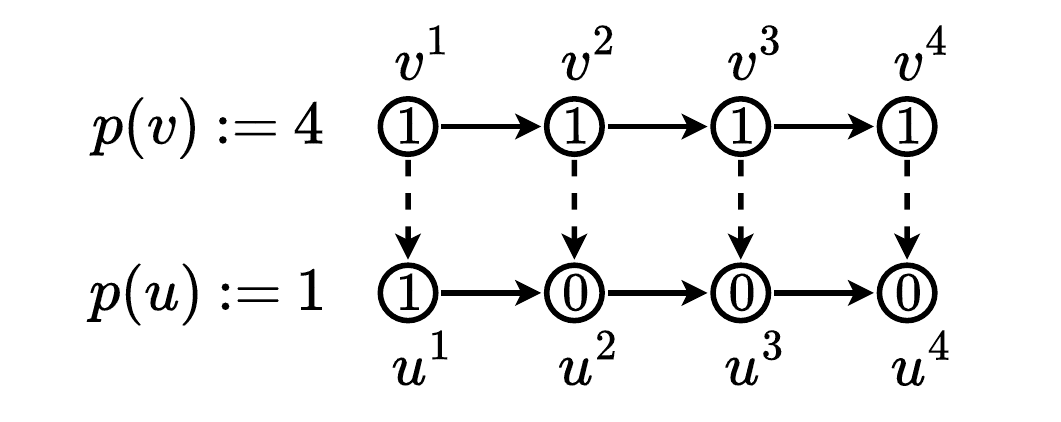}
    \caption{Illustration of the arcs in $\hat{A}_1$ (solid arrows) and ones in $\hat{A}_2$ (dotted arrows) along with the $0$-$1$ assignment for an arc $(u,v)\in A$. Each circle represents a vertex in $\hat{V}$, with the integer inside the circle indicating the $0$-$1$ assignment to it.}
    \label{fig:via-mincut}
\end{figure}

The following lemma verifies the above argument.
Here, let $\mathcal{X} \coloneqq \{ X \subseteq \hat{V} \mid s \in X \not\ni t \}$,
and let $\delta \colon \mathcal{X} \to \Z_+$ be the $s$-$t$ cut function of $\hat{G}$,
i.e.,
$\delta(X) \coloneqq \sum_{a \in \Delta^+(X)} c(a)$.

\begin{lemma}\label{lem:potential-to-mincut}
For any $X \in \mathcal{X}$ with $\delta(X) < M$,
there exists a $k$-potential $p$ such that $H(p) = \delta(X)$,
and this $p$ can be constructed in $O(k|V|)$ time.
  Conversely, for any $k$-potential $p$, there exists $X \in \mathcal{X}$ such that $H(p) = \delta(X)$.
\end{lemma}
\begin{proof}
Take any $X \in \mathcal{X}$ with $\delta(X) < M$.
  Define an integer assignment $p \colon V \to \mathbb{Z}_+$ by $p(v) = \max \{ i \in [0,k] \mid i = 0 \text{ or } v^{i} \notin X \}$ for each $v\in V$.
  First, we show that $p$ is a $k$-potential.
  Since $\delta(X) < M$, we have $\bot^i \notin X \ni \top^i$ for all $i \in [k]$.
  Hence $p$ satisfies~\eqref{cond:P1}.
  The condition~\eqref{cond:P2} is obviously satisfied by the definition of $p$.
  To prove~\eqref{cond:P3}, we show the contraposition.
  Assume that there exists an arc $(u,v) \in A$ with $p(u) > p(v)$.
  For such an arc, we have $u^{p(u)} \notin X \ni v^{p(u)}$ by the definition of $p$, which contradicts that $\delta(X) < M$ as $(v^{p(u)},u^{p(u)}) \in \hat{A}_2$.
  Therefore, $p$ satisfies~\eqref{cond:P3}, and thus, $p$ is a $k$-potential.
  Moreover, since arcs in $\hat{A}_1$ have capacity $M$, it follows that $v^1, v^2, \ldots, v^{p(v)} \notin X$ and $v^{p(v)+1}, v^{p(v)+2}, \ldots, v^k \in X$ for each vertex $v \in V$.
  We can then see that 
  $\delta(X)$ is equal to
  \begin{align}
    \sum_{ (u^i,v^j) \in \hat{A}_3 \cap \Delta^+(X)} w((u,v))
    &= \sum_{ a=(u,v)\in A} \left( w(a)\cdot |\{(u^i,v^j) \in \hat{A}_3 \mid (u^i, v^j) \in \Delta^+(X) \}| \right) \\
    &= \sum_{ a=(u,v) \in A} w(a) \varphi(p(v)-p(u)) \\
    &= H(p). \label{eq:potential-to-mincut-1}
  \end{align}
  Here, the second equation requires a separate discussion for each case of $\varphi(x) = \binom{x}{2}$ and $\varphi(x) = \max \{0,x-1\}$.
  Therefore, we have $H(p) = \delta(X)$.
  Clearly, $p$ can be constructed from $\Delta^+(X)$ in $O(|\hat{V}|) = O(k|V|)$ time.

  Conversely, let $p$ be an arbitrary $k$-potential.
  Define $X \subseteq \hat{V}$ as $X \coloneqq \{ v^i \mid v \in V,\ p(v) < i \leq k \} \cup \{s\}$.
  Clearly, we have $X \in \mathcal{X}$.
  For each $(u,v) \in A$, since $p(u) \leq p(v)$, 
  if $v^i \in X$ for some $i \in [k]$, then $u^i \in X$ also holds.
  This implies that $\Delta^+(X)$ contains no arcs from $\hat{A}_2$.
  Furthermore, $\Delta^+(X)$ contains no arcs from $\hat{A}_1$ by the definition of $X$.
  Hence, we have $\Delta^+(X) \subseteq \hat{A}_3$.
  By reversing the argument
  in~\eqref{eq:potential-to-mincut-1},
  we obtain
  $\delta(X) = H(p)$.
  This completes the proof.
\end{proof}

\Cref{lem:potential-to-mincut} implies that
the minimum $k$-potential problem can be solved in
$O(T_\textup{cut}(k|V|, k^2|A|,M))$ time if $\varphi(x) = \binom{x}{2}$,
and in $O(T_\textup{cut}(k|V|, k|A|,M))$ time if $\varphi(x) = \max\{0,x-1\}$.

Applying the state-of-the-art algorithm~\cite{brand2023} for the maximum $s$-$t$ flow problem to the instance of the minimum $k$-potential problem reduced from an instance of the \SumkDiverse/\CovkDiverse of {\sc Prob},
we obtain the following result.

\begin{theorem}\label{thm:framework-maxflow-time-complexity}
    The problems \SumkDiverse and \CovkDiverse of {\sc Prob} can be solved in $O(T_P(\mathbf{I}) 
+ T_r(\mathbf{I}) + (k|A_P| + k^2|E|)^{1 + o(1)})$ time
    and in $O(T_P(\mathbf{I}) + T_r(\mathbf{I}) + (k(|A_P| + |E|))^{1 + o(1)})$ time,
    respectively,
    where $\mathbf{I}$ is a given instance of {\sc Prob}.
\end{theorem}
\begin{proof}
Recall the notation in \Cref{sec:reduction}.
We have $|P| \leq |A_P| + |E|$,
where a DAG $(P, A_P)$ represents the poset $P$,
and $\sum_{a \in A} w(a) = |E|$.
Thus, the sizes of each arc sets are $|\hat{A}_1| = O(k(|A_P| + |E|))$, $|\hat{A}_2| = O(k(|A_P| + |E|))$, $|\hat{A}_3| = O(k^2 |E|)$ when $\varphi(x) = \binom{x}{2}$, and $|\hat{A}_3| = O(k |E|)$ when $\varphi(x) = \max \{ 0, x-1 \}$.
If we set $M \coloneqq \sum_{a\in A} w(a) \varphi(k) + 1$,
then we have $M = \binom{k}{2}|E| + 1$ when $\varphi(x) = \binom{x}{2}$,
which is $O(k^2|E|)$,
and $M = (k-1)|E| + 1$ when $\varphi(x) = \max \{ 0, x-1 \}$,
which is $O(k|E|)$.
By
\Cref{thm:sec2:k-diverse}
and
$T_\mathrm{cut}(n, m, C) = O(m^{1+o(1)} \log C)$,
we obtain
$O(T_P(\mathbf{I}) + T_r(\mathbf{I}) + (k|A_P| + k^2|E|)^{1 + o(1)})$ time
for \SumkDiverse,
and $O(T_P(\mathbf{I}) + T_r(\mathbf{I}) + (k(|A_P| + |E|))^{1 + o(1)})$ time for \CovkDiverse.
\end{proof}

At this moment, even if we use the state-of-the-art algorithm for the maximum $s$-$t$ flow problem, the time bound in \Cref{thm:framework-maxflow-time-complexity} is worse than that in \Cref{thm:optimal-potential-time-complexity}, derived from the reduction to the minimum cost flow problem.
Nevertheless, we believe that
it is still worth considering our reduction to the maximum $s$-$t$ flow problem for practical use.
Recently developed algorithms for the network flow problems are not easy to implement.
    Since the reduction procedure and the reduced problem, namely, the maximum $s$-$t$ flow problem, of this subsection are simpler than those of the previous subsection,
    the reduction in this subsection may perform faster in practice 
    if we use classical (or implementable) algorithms.

\section{Applications}\label{sec:applications}
In this section,
we introduce two applications of our framework;
one is the $k$-diverse problems of \Mincut
and the other is that of \StableMatching.

In order to apply our framework to the $k$-diverse problem of a concrete combinatorial problem,
we need to construct a poset and a reduction map appearing in the property~\eqref{cond:R} for each instance.
In \Cref{subsec:strategy},
we develop a common strategy of building these components.
In fact, it is known that both families of minimum $s$-$t$ cuts and stable matchings are naturally identified with set families called \emph{ring families}.
A ring family with inclusion order forms a poset, particularly a \emph{distributive lattice}.
We will utilize these facts to construct a reduction map for \Mincut
and \StableMatching in \Cref{subsec:k-diverse-mincut,subsec:k-diverse-stable-matching}, respectively.
In \Cref{subsec:k-diverse-totalorder},
we briefly describe a framework for the $k$-diverse problem recently developed by De Berg, Mart\'{i}nez, and Spieksma in~\cite{deBerg2025},
and show that our framework can capture theirs in the case of \SumkDiverse/\CovkDiverse.

\subsection{How to construct a poset and a reduction map}\label{subsec:strategy}
We start this subsection with introducing terminology on lattice (see e.g.,~\cite{Davey2002} for details).
A poset $(L, \preceq)$ is called a \emph{lattice}
if, for any two elements $x, y \in L$,
their least upper bound and greatest lower bound exist in $L$;
the former and latter are called the \emph{join}
and \emph{meet} of $x$ and $y$,
which are denoted as $x \vee y$ and $x \wedge y$,
respectively.
A lattice $L$ is said to be \emph{distributive} if the distributive law $x \wedge (y \vee z) = (x \wedge y) \vee (x \wedge z)$ holds for any $x, y, z \in L$.

A typical example of a distributive lattice is the family $\cI(P)$ of ideals of a poset $P$ with inclusion order $\subseteq$.
In this case,
the join of two ideals $I$ and $I'$ is their union $I \cup I'$,
and the meet is their intersection $I \cap I'$;
we can easily see that the union and intersection of two ideals are also ideals.
More generally,
a \emph{ring family},
which is a nonempty family $\cR \subseteq 2^R$ of subsets of a nonempty finite set $R$
such that it is closed under the union and intersection,
endowed with inclusion order $\subseteq$
forms a distributive lattice $(\cR, \subseteq)$,
where $X \vee Y = X \cup Y$ and $X \wedge Y = X \cap Y$
for $X, Y \in \cR$.

The celebrated \emph{Birkhoff's representation theorem}~\cite{birkhoff1937} states that
every distributive lattice $(L, \preceq)$ is isomorphic to the distributive lattice $(\cI(P), \subseteq)$ over the family of ideals of some poset $P$.
Here, two lattices $(L, \preceq)$ and $(L', \preceq')$ are said to be \emph{isomorphic}
if there is a bijection $h \colon L \to L'$ such that
$x \preceq y$ if and only if $h(x) \preceq h(y)$ for any $x, y \in L$.
For a distributive lattice $(L, \preceq)$,
we refer to
a poset $P$ such that $(L, \preceq)$ and $(\cI(P), \subseteq)$ are isomorphic as a \emph{Birkhoff representation} of $(L, \preceq)$.

A Birkhoff representation of the distributive lattice $(\cR, \subseteq)$ over a ring family $\cR$ can be obtained as follows.
The ring family $\cR$ has the unique minimal set $X_{\bot} \coloneqq \bigcap_{X \in \cR} X$ and unique maximal set $X_{\top} \coloneqq \bigcup_{X \in \cR} X$.
Take any maximal chain $X_{\bot} \eqqcolon X_0 \subsetneq X_1 \subsetneq \cdots \subsetneq X_n \coloneqq X_{\top}$
from $X_{\bot}$ to $X_{\top}$ in $\cR$;
namely, there are no $X \in \cR$ and $i \in [n]$ with $X_{i-1} \subsetneq X \subsetneq X_i$.
Then, $\Pi^*(\cR) \coloneqq \{ X_i \setminus X_{i-1} \mid i \in [n] \}$ forms a partition of $X_{\top} \setminus X_{\bot}$.
We define a partial order $\preceq$ on $\Pi^*(\cR)$ by setting $\hat{X} \preceq \hat{Y}$ if and only if every $Z \in \cR$
with $Z \supseteq \hat{Y}$
also includes $\hat{X}$.
The resulting $(\Pi^*(\cR), \preceq)$ is actually a poset and is independent of the choice of a maximal chain from $X_\bot$ to $X_\top$.
It is known (see e.g.,~\cite[Chapter~2.2.2]{Murota2010})
that $(\Pi^*(\cR), \preceq)$ is a Birkhoff representation of $(\cR, \subseteq)$; more precisely,
the map $\mathcal{Y} \mapsto (\bigcup_{\hat{X} \in \mathcal{Y}} \hat{X}) \cup X_\bot$
is an isomorphism from $(\cI(\Pi^*(\cR)), \subseteq)$ to $(\cR, \subseteq)$.
In this paper, we refer to $(\Pi^*(\cR), \preceq)$ as \emph{the} Birkhoff representation of $(\cR, \subseteq)$.

We are ready to develop a strategy to construct a reduction map by using a ring family.
Let $\cR \subseteq 2^R$ be a ring family over a nonempty finite set $R$.
We may assume that the minimal set $X_\bot$ is nonempty and the maximal set $X_\top$ is a proper subset of $R$;
otherwise we add two elements $\bot$ and $\top$ to $R$
and update each subset $X \in \cR$ as $X \cup \{ \bot \}$,
which makes $\cR$ satisfy $\bot \in X_\bot$ and $\top \in \overline{X}_\top \coloneqq R \setminus X_\top$.
Let $\Pi(\cR)$ denote the partition $\Pi^*(\cR) \cup \{ X_\bot, \overline{X}_\top \}$ of $R$.
We extend the partial order $\preceq$ on $\Pi^*(\cR)$ to that on $\Pi(\cR)$ by setting
$X_\bot \prec \hat{X} \prec \overline{X}_\top$ for any $\hat{X} \in \Pi^*(\cR)$.
For a map $\hat{r} \colon E \to R^2$ given by $e \mapsto (\hat{e}^+, \hat{e}^-)$,
we define
\begin{align}
    {\sup}_{\hat{r}}(X) \coloneqq \{ e \in E \mid \hat{e}^+ \in X \not\ni \hat{e}^- \}
\end{align}
for $X \in \cR$.
We say that $\hat{r} \colon E \to R^2$ is a \emph{pre-reduction map}
if $\hat{e}^- \in X$ implies $\hat{e}^+ \in X$ for all $X \in \cR$ and $e \in E$, and
$\cS(\mathbf{I}) = \{ {\sup}_{\hat{r}}(X) \mid X \in \cR \}$.

We show that we can construct a reduction map $r \colon E \to \Pi(\cR)^2$ as long as we have a pre-reduction map $\hat{r} \colon E \to R^2$.

\begin{lemma}\label{lem:pre-reduction}
    Suppose that $\hat{r} \colon E \to R^2$ given by $e \mapsto (\hat{e}^+, \hat{e}^-)$ is a pre-reduction map.
    Then, the map $r \colon E \to \Pi(\cR)^2$ defined by
    \begin{align}
        r(e) \coloneqq (\Pi(\hat{e}^+), \Pi(\hat{e}^-)) \qquad (e \in E)
    \end{align}
    is a reduction map, where $\Pi(x)$ ($x\in R$) denotes the unique member of $\Pi(\cR)$ that contains $x$.
\end{lemma}
\begin{proof}
It suffices to see that (i)
for all $e \in E$, we have $\Pi(\hat{e}^+) \preceq \Pi(\hat{e}^-)$
and (ii) $\cS(\mathbf{I}) = \{ {\sup}_r(I) \mid I \in \cI(\Pi^*(\cR)) \}$.
For each $I \in \cI(\Pi^*(\cR))$,
let $\bigcup I \coloneqq \bigcup_{\hat{X} \in I} \hat{X}$.
Recall that the map $I \mapsto (\bigcup I) \cup X_\bot$
is an isomorphism from $(\cI(\Pi^*(\cR)), \subseteq)$ to $(\cR, \subseteq)$.
Hence, every $X \in \cR$ is representable as the union of several members of $\Pi(\cR)$,
which implies that, for $x \in R$ and $X \in \cR$,
$x \in X$ if and only if $\Pi(x) \subseteq X$.

We first see (i). By the first condition of a pre-reduction map,
$\Pi(\hat{e}^-) \subseteq X$ implies $\Pi(\hat{e}^+) \subseteq X$
for all $X \in \cR$ and $e \in E$.
Hence, we obtain $\Pi(\hat{e}^+) \preceq \Pi(\hat{e}^-)$ by the definition of the partial order $\preceq$ on $\Pi(\cR)$.

Next, (ii) follows because 
\begin{align}
    \{{\sup}_r(I) \mid I \in \cI(\Pi^*(\cR)) \} &= \left\{ \{e \in E \mid \Pi(\hat{e}^+) \in I \cup \{ X_\bot \} \not\ni \Pi(\hat{e}^-) \} \ \middle|\  I \in \cI(\Pi^*(\cR)) \right\}\\
    &= \left\{ \left\{e \in E \ \middle|\  \Pi(\hat{e}^+) \subseteq \left(\bigcup I\right) \cup X_\bot \not\supseteq \Pi(\hat{e}^-) \right\} \ \middle|\ I \in \cI(\Pi^*(\cR)) \right\}\\
    &= \left\{ \{e \in E \mid \hat{e}^+ \in X \not\ni \hat{e}^- \} \mid X \in \cR \right\}\\
    &= \{ {\sup}_{\hat{r}}(X) \mid X \in \cR \}\\
    &= \cS(\mathbf{I}).
\end{align}
Here, the third equality holds because the map $I \mapsto (\bigcup I) \cup X_\bot$ is a bijection between $\cI(\Pi^*(\cR))$ and $\cR$, and 
$x \in X$ if and only if $\Pi(x) \subseteq X$.
The last equality follows from the second condition of a pre-reduction map.
\end{proof}

Thus, we obtain the following:
\begin{proposition}\label{prop:pre-reduction}
    Let $R$ be a nonempty finite set,
    and let $\cR$ be a ring family over $R$.
    Suppose we are given a pre-reduction map $\hat{r} \colon E \to R^2$
    and the partition $\Pi(\cR)$ of $R$.
    Then, we can construct a reduction map $r \colon E \to \Pi(\cR)^2$ in constant time.
\end{proposition}

In \Cref{subsec:k-diverse-mincut,subsec:k-diverse-stable-matching},
we actually construct pre-reduction maps
by using ring families that are closely related to their solution sets.

The property~\eqref{cond:R} requires that, for each instance $\mathbf{I}$ of {\sc Prob},
the family $\cS(\mathbf{I})$ of its solutions is closely related to the family $\cI(P^*)$ of ideals of the poset $P^*$.
Hence, it is natural to expect that $\cS(\mathbf{I})$ admits the distributive lattice structure.
We conclude this subsection by showing that this expectation is correct,
which might be viewed as
a necessary condition of a combinatorial problem for our framework to be applicable.

Recall the notation in the property~\eqref{cond:R}.
We may assume that $e^+ \prec e^-$ for each $e \in E$,
since if $e^+=e^-$, then no member $S$ in $\cS(\mathbf{I})$ contains $e$, i.e., $\cS(\mathbf{I}) \subseteq \binom{E \setminus \{e\}}{q}$, and 
hence we can remove $e$ from $E$.
Let $\leq$ denote a partial order on $E$ defined by setting
$e \leq e'$ if and only if $e = e'$ or $e^- \preceq e'^+$.
This $\leq$ can be extended to a partial order on $\cS(\mathbf{I})$ by setting
$S \leq T$ if and only if
there exists a bijection $\pi \colon S \to T$ such that
$e \leq \pi(e)$ for all $e \in S$.
Then, the following holds.

\begin{theorem}
    Suppose that a combinatorial problem {\sc Prob} has the properties~\eqref{cond:S} and~\eqref{cond:R}.
    For any instance $\mathbf{I}$ of {\sc Prob}, 
    the poset $(\cS(\mathbf{I}), \leq)$ forms a distributive lattice.
\end{theorem}
\begin{proof}
Let $(P,\preceq)$ be the poset associated with an instance $\mathbf{I}$ indicated in the property~\eqref{cond:R}.
For each $x \in P$,
we define
\begin{align}
    \Theta^+(x) \coloneqq \{ e \in E \mid e^+ = x \}, \quad \Theta^-(x) \coloneqq \{ e \in E \mid e^- = x \}, \quad \theta(x) \coloneqq |\Theta^+(x)| - |\Theta^-(x)|.
\end{align}
We first observe the following claim.
\begin{claim}\label{cl:obs}
    \begin{enumerate}[{label={\textup{(\arabic*)}}}]
        \item
        $\theta(\bot) = q$, $\theta(\top) = -q$, and $\theta(x) = 0$ for $x \in P^*$.
        \item ${\sup}_r(I) = \bigcup_{x \in I \cup \{\bot\}} \Theta^+(x) \setminus \bigcup_{x \in I \cup \{\bot\}} \Theta^-(x)$
        for each $I \in \cI(P^*)$.
        \item Let $I \in \cI(P^*)$ and $x \in I \cup \{ \bot \}$ be a maximal element in $I \cup \{ \bot \}$.
        Then we have ${\sup}_r(I) \supseteq \Theta^+(x)$.
    \end{enumerate}
\end{claim}
\begin{proof}[Proof of \Cref{cl:obs}]
    (1).
    By the assumption that $e^+ \prec e^-$ for $e \in E$,
    there is no $e \in E$ such that $e^- = \bot$ or $e^+ = \top$.
    Hence, $\theta(\bot) = |\Theta^+(\bot)| = |{\sup}_r(\{\bot\})| = q$
    and $\theta(\top) = - |\Theta^-(\top)| = - |{\sup}_r(P \setminus \{ \top \})| = -q$.
    Suppose, to the contrary, that there is $x \in P^*$ with $\theta(x) \neq 0$.
    Let $I \in \cI(P^*)$ be an ideal of $P^*$ such that $x \notin I$ and $I \cup \{x\}$ is also an ideal.
    Then we have $|{\sup}_r(I \cup \{x\})| = |{\sup}_r(I)| + \theta(x)$,
    which implies $|{\sup}_r(I \cup \{x\})| \neq |{\sup}_r(I)|$.
    This contradicts the property~\eqref{cond:S}.

    (2). This assertion follows from
    ${\sup}_r(I) = \{ e \in E \mid e^+ \in I \cup \{\bot\} \not\ni e^- \}
    = \{ e \in E \mid e^+ \in I \cup \{\bot\} \} \setminus \{ e \in E \mid e^- \in I \cup \{\bot\} \}$.

    (3). Since $x$ is maximal in $I \cup \{ \bot \}$ and $e^+ \prec e^-$ by the assumption,
    for any $e \in E$ with $e^+ = x$
    we have $e^- \notin I \cup \{\bot\}$.
    Hence, we obtain ${\sup}_r(I) \supseteq \Theta^+(x)$ by~(2).
\end{proof}

Here we may assume that
$\Theta^+(x) \cup \Theta^-(x) \neq \emptyset$
for each $x \in P$,
since otherwise we can replace $P$ with the subposet of $P$ induced by $\{ x \in P \mid \Theta^+(x) \cup \Theta^-(x) \neq \emptyset \}$.
In the following, we prove that ${\sup}_r \colon \cI(P) \to \cS(\mathbf{I})$ is an isomorphism between $(\cS(\mathbf{I}), \leq)$ and $(\cI(P^*), \subseteq)$,
which implies that $(\cS(\mathbf{I}), \leq)$ forms a distributive lattice.

We first see that ${\sup}_r$ is bijective and that ${\sup}_r(I) \not\leq {\sup}_r(I')$ holds for each $I, I' \in \cI(P^*)$ with $I \not\subseteq I'$.
By the property~\eqref{cond:R}, ${\sup}_r$ is surjective.
Take any distinct $I, I' \in \cI(P^*)$ such that $I \not\subseteq I'$, i.e., $I \setminus I' \neq \emptyset$.
Take any maximal element $x$ in $I \setminus I'$; note that this $x$ is also maximal in $I$.
Then we have ${\sup}_r(I) \supseteq \Theta^+(x) \neq \emptyset$ by \Cref{cl:obs}~(3).
Furthermore, for each $e \in \Theta^+(x)$ and $e' \in {\sup}_r(I')$,
we have $e \not\leq e'$,
since otherwise we have $I' \ni e'^+ \succeq e^+ = x$, which contradicts that $x \notin I'$.
Hence, there is no bijection $\pi \colon {\sup}_r(I) \to {\sup}_r(I')$ satisfying that
$e \leq \pi(e)$ for each $e \in E$.
This means that ${\sup}_r(I) \not\leq {\sup}_r(I')$,
which also implies the injectivity of ${\sup}_r$.

We finally prove that for each $I, I' \in \cI(P^*)$ with $I \subseteq I'$,
we have ${\sup}_r(I) \leq {\sup}_r(I')$.
It suffices to consider the case where $I' = I \cup \{x\}$ for some $x \notin I$.
In this case,
we have ${\sup}_r(I') = ({\sup}_r(I)\setminus \Theta^-(x)) \cup \Theta^+(x)$ by \Cref{cl:obs}~(2).
Since $|\Theta^-(x)| = |\Theta^+(x)|$ by \Cref{cl:obs}~(1),
there is a bijection $\pi' \colon\Theta^-(x) \to \Theta^+(x)$.
For any $e \in \Theta^-(x)$ and $e' \in \Theta^+(x)$,
since $e^+ = x = e'^-$,
we have $e \leq e'$.
Thus, the map $\pi \colon {\sup}_r(I) \to {\sup}_r(I')$ defined by
\begin{align}
\pi(e) \coloneqq
\begin{cases}
    e & \text{if $e \notin \Theta^-(x)$},\\
    \pi'(e) & \text{if $e \in \Theta^-(x)$}
\end{cases}
\end{align}
for each $e\in {\sup}_r(I)$
is a bijection satisfying
$e \leq \pi(e)$ for all $e \in {\sup}_r(I)$.
Therefore, we obtain ${\sup}_r(I) \leq {\sup}_r(I')$.
\end{proof}

\subsection{The \texorpdfstring{$k$}{k}-diverse unweighted minimum \texorpdfstring{$s$-$t$}{s-t} cut problem}\label{subsec:k-diverse-mincut}
In this subsection,
we consider the $k$-diverse problem of \Mincut.
Here, \Mincut is the minimum $s$-$t$ cut problem for a digraph $G = (V, A)$ with unit arc capacity $c \colon A \to \Z_+$, i.e., $c(a) = 1$ for all $a \in A$.
In other words,
we are asked to find an $s$-$t$ cut $C \subseteq A$ with minimum size $|C|$.
Clearly, \Mincut has the property~\eqref{cond:S} by setting the ground set $E$ as the arc set $A$ and the integer $q$ as the size of minimum $s$-$t$ cuts.
Our aim is to construct,
for some ring family $\cR$,
a DAG $D$ representing the Birkhoff representation $(\Pi^*(\cR), \preceq)$,
the partition $\Pi(\cR)$,
and
a pre-reduction map $\hat{r}$ based on $\cR$,
which implies that \Mincut has the property~\eqref{cond:R} by \Cref{lem:pre-reduction}.

Let $G = (V, A)$ with $s, t \in V$ be a digraph that is an instance of \Mincut.
Then, the family $\cS(G) \subseteq 2^A$ of solutions of the instance $G$ is the set of all minimum $s$-$t$ cuts,
which is a subset of the family $\{ \Delta^+(X) \mid \text{$X \subseteq V$ with $s \in X \not\ni t$} \}$.
Let $\cR_{st}$ denote the family of vertex subsets $X \subseteq V$ with $s \in X \not\ni t$ such that its outgoing arcs form a minimum $s$-$t$ cut of $G$,
i.e.,
\begin{align}
    \cR_{st} \coloneqq \{ X \subseteq V \mid \text{$s \in X \not\ni t$, $\Delta^+(X)$ is a minimum $s$-$t$ cut} \}.
\end{align}
It is well-known that $\cR_{st}$ forms a ring family, i.e., for any $X, Y \in \cR_{st}$,
we have $X \cup Y, X \cap Y \in \cR_{st}$,
which directly follows from \Cref{lem:mincut-characterization} below.
Moreover,
since the unique minimal set $X_{\bot}$ in $\cR_{st}$ contains $s$ and the unique maximal set $X_{\top}$ in $\cR_{st}$ excludes $t$,
we have $X_{\bot} \neq \emptyset\neq \overline{X}_{\top} \coloneq V \setminus X_{\top}$.
We will construct a pre-reduction map based on this $\cR_{st}$.

We utilize the representation introduced by Picard and Queyranne~\cite{picard1980}, which is used for enumerating all minimum $s$-$t$ cuts of a digraph with arbitrary positive capacities.
In our case, we consider the unit capacity case, and use this representation to construct a DAG representing the Birkhoff representation $(\Pi^*(\cR_{st}), \preceq)$ of the distributive lattice $(\cR_{st}, \subseteq)$, as well as the partition $\Pi(\cR_{st})$ of $V$.

To begin, we briefly describe the representation proposed by Picard and Queyranne.
Let $f$ be an arbitrary maximum $s$-$t$ flow of $G$, where each arc has unit capacity.
Let $G_f$ denote the residual graph of $f$.
Then, the following characterization is known.

\begin{lemma}[\cite{picard1980}]\label{lem:mincut-characterization}
  An $s$-$t$ cut $\Delta^+(X)$ is a minimum $s$-$t$ cut in $G$ if and only if $X$ is a closed set under reachability in $G_f$, containing $s$ but not $t$.
\end{lemma}

Using this characterization, Picard and Queyranne constructed a DAG $D=(V',A')$ from $G$ through the following steps:
(1) Contract each strongly connected component in $G_f$ into a single vertex,
(2) remove the component containing $s$ along with all vertices reachable from $s$, and
(3) remove the component containing $t$ along with all vertices that can reach $t$.
For each vertex $v'\in V'$, let $R(v')$ denote the set of vertices in $G_f$ that are contracted into $v'$.
Similarly, let $R(s')$ and $R(t')$ denote the sets of vertices in $G_f$ that are removed in steps (2) and (3), respectively.
By \Cref{lem:mincut-characterization}, the DAG $D$ represents all minimum $s$-$t$ cuts in $G$.
Specifically, for any $X \subseteq V'$ that is closed under reachability in $D$, the set of outgoing arcs from $R(X) \coloneqq (\bigcup_{v'\in X} R(v')) \cup R(s')$ in $G$ forms a minimum $s$-$t$ cut in $G$.
Conversely, for any minimum $s$-$t$ cut $\Delta^+(X)$ in $G$, the set of components corresponding to $X \setminus R(s')$ is closed under reachability in $D$.

In fact, $D$ serves as the desired DAG that represents the Birkhoff representation $(\Pi^*(\cR_{st}), \preceq)$ of $\cR_{st}$.
Let $(V', \preceq)$ be a poset where reachability in $D$ defines the partial order i.e., for any $u,v \in V'$, $u \preceq v$ if and only if $u$ is reachable from $v$.
From the above discussion, the map $R \colon \cI(V') \to \cR_{st}$ is an isomorphism between $(\cI(V'), \subseteq)$ and $(\cR_{st}, \subseteq)$, because for all $X, X' \in \cI(V')$, we have $X \subseteq X'$ if and only if $R(X) \subseteq R(X')$.
Moreover, choose an arbitrary maximal chain $\emptyset = X_0 \subsetneq X_1 \subsetneq \cdots \subsetneq X_n = V'$ in $\cI(V')$.
Since $R(X_0) \subsetneq R(X_1) \subsetneq \cdots \subsetneq R(X_n)$ forms a maximal chain from $X_\bot$ to $X_\top$ in $\cR_{st}$, and 
$|X_i \setminus X_{i-1}| = 1$ for all $i \in [n]$, 
we obtain an isomorphism that maps each $v \in V'$ to $R(v)$, establishing the equivalence between $(V', \preceq)$ and $(\Pi^*(\cR_{st}), \preceq)$.

Thus, as all strongly connected components in $G$ can be found in linear time by~\cite{sharir1981}, the partition $\Pi(\cR_{st})$ and a DAG representing $(\Pi^*(\cR_{st}), \preceq)$ can be constructed from $G_f$ in $O(|V| + |A|)$ time.
Furthermore,
the number of vertices and arcs representing $(\Pi^*(\cR_{st}), \preceq)$ are at most $|V|$ and $|A|$,
respectively.

We then define a map $\hat{r} \colon A \to V^2$ by
\begin{align}\label{eq:mincut-hatr}
    \hat{r}(a) \coloneqq
    \begin{cases}
        a & \text{if $f(a) = 1$},\\
        (t, t) & \text{if $f(a) = 0$}.
    \end{cases}
\end{align}
We confirm that $\hat{r}$ is a pre-reduction map.
This together with \Cref{prop:pre-reduction} and the above argument implies that $T_P(G) + T_r(G) = O(|V| + |A|)$.
\begin{lemma}\label{lem:mincut-pre-reduction}
    The map $\hat{r} \colon A \to V^2$ defined as~\eqref{eq:mincut-hatr} is a pre-reduction map.
\end{lemma}
\begin{proof}
We denote $(\hat{a}^+, \hat{a}^-)=\hat{r}(a)$ for each $a\in A$.
It suffices to see that
(i) for each $X \in \cR_{st}$ and arc $a \in A$,
if $\hat{a}^- \in X$ then $\hat{a}^+ \in X$,
and (ii) $\cS(G) = \{ \sup_{\hat{r}}(X) \mid X \in \cR_{st}\}$.

To see (i), take any $X \in \cR_{st}$ and arc $a = (u,v) \in A$ with $\hat{a}^- \in X$.
Since $t \notin X_{\top}$ and $X \subseteq X_{\top}$,
we have $\hat{a}^- \neq t$,
which implies that $(\hat{a}^+, \hat{a}^-) = a = (u, v)$ and $f(a) = 1$.
Therefore, the reverse arc $(v, u)$ of $a$ exists in $G_f$.
Since $X$ contains $\hat{a}^- = v$ and $X$ is closed under the reachability in $G_f$ by \Cref{lem:mincut-characterization},
we obtain $\hat{a}^+ = u \in X$.

(ii) follows from the definition of $\cS(G)$, namely $\cS(G) = \{ \Delta^+(X) \mid X \in \cR_{st}
 \}$, and a basic fact, which is also directly implied by~\Cref{lem:mincut-characterization}, that for any $X \in \cR_{st}$, all outgoing arcs $a \in \Delta^+(X)$ from $X$ satisfy $f(a) = 1$.
\end{proof}

Thus, by \Cref{thm:optimal-potential-time-complexity},
we obtain the following, which implies \Cref{thm:mincut-stablematching}~(1).

\begin{theorem}
    The $k$-diverse problem with respect to $d_\varphi$ of \Mincut can be solved in 
    $O(n + (|B_k(\varphi)|m)^{1+o(1)} \log (\varphi(k)) \log k + kq)$ time,
    where $n$ denotes the number of vertices in the input digraph $G$,
    $m$ the number of arcs,
    and
    $q$ the size of any minimum $s$-$t$ cut in $G$.
    In particular,
    \SumkDiverse and \CovkDiverse of \Mincut can be solved in
    $O(n + (km)^{1 + o(1)})$ time
    and in
    $O(n + m^{1 + o(1)} \log^2 k + kq)$ time,
    respectively.
\end{theorem}
\begin{proof}
Clearly, we have $|E| = m$.
Since
$T_P(G) + T_r(G) = O(n + m)$
and $|A_P| = O(m)$ hold as we have seen above,
we obtain
$O(n + (|B_k(\varphi)|m)^{1+o(1)} \log (\varphi(k)) \log k + kq)$ time
for the $k$-diverse with respect to $d_{\varphi}$,
$O(n + (km)^{1 + o(1)})$ time
for \SumkDiverse,
and $O(n + m^{1 + o(1)} \log^2 k + kq)$ time
for \CovkDiverse
by \Cref{thm:optimal-potential-time-complexity}.
\end{proof}

\subsection{The \texorpdfstring{$k$}{k}-diverse stable matching problem}\label{subsec:k-diverse-stable-matching}
In this subsection,
we consider the $k$-diverse problem of \StableMatching.
Let us first introduce the problem \StableMatching (see e.g.,~\cite{gusfield1989} for details).
Let $U$ and $V$ be disjoint finite sets with the same size $n$
endowed with total orders $\leq_u$ on $V$ ($u \in U$) and $\leq_v$ on $U$ ($v \in V$).
Intuitively, the total order $\leq_u$ (resp.\ $\leq_v$) represents the preference of $u$ on $V$ (resp.\ $v$ on $U$);
$v <_u v'$ means that ``$u$ prefer $v$ to $v'$.''
A subset $M = \{(u_1, v_1), (u_2, v_2), \dots, (u_n, v_n)\} \subseteq U \times V$ is called a \emph{matching}
if $M$ provides the one-to-one correspondence between $U$ and $V$,
i.e., all $u_i, v_j$ are different, $U = \{ u_1, u_2, \dots, u_n \}$, and $V = \{ v_1, v_2, \dots, v_n \}$.
For a matching $M$,
we denote by $p_M(u)$ (resp.~$p_M(v)$) the partner of $u \in U$ (resp.~$v \in V$) in $M$,
i.e.,
$(u, p_M(u)) \in M$ (resp. $(p_M(v), v) \in M$).
A matching $M$ is said to be \emph{stable} in $(U, V; (\leq_u)_{u \in U}, (\leq_v)_{v \in V})$
if there is no pair $(u,v) \in U \times V$ such that $v <_u p_M(u)$ and $u <_v p_M(v)$.
In the problem \StableMatching,
we are given a tuple $(U, V; (\leq_u)_{u \in U}, (\leq_v)_{v \in V})$ of finite sets $U, V$ with the same size $n$ and the total orders $(\leq_u)_{u \in U}$ on $V$ and $(\leq_v)_{v \in V}$ on $U$,
and asked to find a stable matching $M$ of $(U, V; (\leq_u)_{u \in U}, (\leq_v)_{v \in V})$.
A stable matching always exists and can be found in $O(n^2)$ time by the \emph{Gale--Shapley algorithm}~\cite{gale1962}.

Clearly, \StableMatching has the property~\eqref{cond:S}, in which $q = |U| = |V|$.
Similarly to \Cref{subsec:k-diverse-mincut},
our aim is to construct,
for some ring family $\cR$,
a DAG representing the Birkhoff representation $(\Pi^*(\cR), \preceq)$,
the partition $\Pi(\cR)$,
and
a pre-reduction map $\hat{r}$ based on $\cR$.

Let $\mathbf{I} \coloneqq (U, V; (\leq_u)_{u \in U}, (\leq_v)_{v \in V})$ be an instance of \StableMatching, and let $n \coloneqq |U| = |V|$.
Then, the family $\cS(\mathbf{I})$ of solutions of the instance $\mathbf{I}$ is the set of all stable matchings of $\mathbf{I}$.
It is known~\cite[Theorem~1.3.2]{gusfield1989} that $\cS(\mathbf{I})$ forms a distributive lattice with the partial order $\preceq$ defined by setting
$M \preceq M'$ if and only if $p_M(u) \leq_u p_{M'}(u)$ for all $u\in U$.
For each stable matching $M \in \cS(\mathbf{I})$, the \emph{P-set} of $M$, denoted by $P(M)$, is the set of pairs $(u,v)$ such that $v$ is at least as preferred as $p_M(u)$ by $u \in U$, i.e., 
$P(M) \coloneqq \{ (u, v) \in U \times V \mid v \leq_u p_M(u) \}$.
Let $\cR_{\mathbf{I}}$ denote the family of all P-sets for $\mathbf{I}$,
i.e.,
\begin{align}
    \cR_{\mathbf{I}} \coloneqq \{ P(M) \mid \text{$M$ is a stable matching} \}.
\end{align}
The map that sends a stable matching to its P-set is an isomorphism between $(S(\mathbf{I}), \preceq)$ and $(\cR_{\mathbf{I}}, \subseteq)$.
Hence, $\cR_{\mathbf{I}}$ forms a ring family, which serves the base of our pre-reduction map.
We note that the unique minimal set $X_\bot$ of $\cR_{\mathbf{I}}$ is nonempty.
Since the unique maximal set $X_{\top}$ could be the whole set $U \times V$,
we add a new element $\top$ to $V$.
In the following,
we regard $U \times (V \cup \{\top\})$ as the ground set so that $X_{\top}$ is a proper subset of the ground set.

We can construct the Birkhoff representation $(\Pi^*(\cR_{\mathbf{I}}), \preceq)$ of the distributive lattice $(\cR_{\mathbf{I}}, \subseteq)$ following the work of Gusfield and Irving~\cite[Chapter~3]{gusfield1989}.
For each stable matching $M$ and each $u\in U$, let $s_M(u)$ denote the most preferable element $v \in V$ for $u$ such that $u <_v p_M(v)$, if such a $v$ exists.
Gusfield and Irving introduced the notion of a rotation, which is an ordered list $\rho = ((u_0,v_0), (u_1,v_1), \ldots, (u_{c-1},v_{c-1}))$ of pairs in some stable matching $M$, satisfying $s_M(u_i) = v_{i+1}$ for all $i \in [0,c-1]$ (where $i+1$ is taken modulo $c$).
Let $\Lambda_\mathbf{I}$ be the set of all rotations of $\mathbf{I}$.
It is shown that the map $d$ defined by
\begin{equation}
    d(\rho) = \{ (u_i,v) \in U \times V \mid i \in [0,c-1], v_i <_{u_i} v \leq_{u_i} v_{i+1} \}
\end{equation}
is a bijection between $\Lambda_\mathbf{I}$ and $\Pi^*(\cR_\mathbf{I})$.
A partial order on $\Lambda_\mathbf{I}$ is induced by $\Pi^*(\cR_\mathbf{I})$ through $d$.
Gusfield and Irving present an $O(n^2)$-time algorithm that constructs $\Lambda_\mathbf{I}$ along with a DAG representing the poset over $\Lambda_\mathbf{I}$ having $O(n^2)$ arcs~\cite[Lemma~3.3.2 and Corollary~3.3.1]{gusfield1989}.
In particular, their algorithm computes the unique minimal and maximal stable matchings.

Since $\Pi^*(\cR_\mathbf{I})$ is a partition of $X_\top \setminus X_\bot$, we can construct $\Pi^*(\cR_\mathbf{I})$ from $\Lambda_\mathbf{I}$ in $O(n^2)$ time using $d$.
Clearly, $X_\bot$ is the P-set of the unique minimal stable matching.
Therefore, we can construct the partition $\Pi(\cR_\mathbf{I})$ and a DAG with $O(n^2)$ arcs representing $(\Pi^*(\cR_\mathbf{I}), \preceq)$ in $O(n^2)$ time.

We then construct a pre-reduction map.
For each $u \in U$ and $v \in V$,
we denote by $v_u^-$ the element in $V$ that is the cover of $v$ with respect to $\leq_u$,
i.e., $v <_u v_u^-$ and there is no element $v'$ in $V$ with $v <_u v' <_u v_u^-$.
If $v$ is the maximum element with respect to $\leq_u$, then we define $v_u^- \coloneqq \top$.
We define a map $\hat{r} \colon U \times V \to \left(U \times (V \cup \{ \top \})\right)^2$
by
\begin{align}\label{eq:stablematching-hatr}
    \hat{r}(u, v) \coloneqq \left((u, v), (u, v_u^-)\right).
\end{align}
We show that this is a pre-reduction map.
This together with \Cref{prop:pre-reduction} and the above argument implies that we have $T_P(\mathbf{I}) + T_r(\mathbf{I}) = O(n^2)$.

\begin{lemma}\label{lem:stablematching-pre-reduction}
    The map $\hat{r} \colon U \times V \to \left(U \times (V \cup \{ \top \})\right)^2$ defined as~\eqref{eq:stablematching-hatr} is a pre-reduction map.
\end{lemma}
\begin{proof}
It suffices to see that
(i) for each $X \in \cR_{\mathbf{I}}$ and pair $(u,v) \in U \times V$,
if $(u, v_u^-) \in X$ then $(u,v) \in X$
and (ii) $\cS(\mathbf{I}) = \{ \sup_{\hat{r}}(X) \mid X \in \cR_{\mathbf{I}}\}$.

(i). Take any $X \in \cR_{\mathbf{I}}$ and $(u, v) \in U \times V$ with $(u, v_u^-) \in X$.
Let $M \in \cS(\mathbf{I})$ be a stable matching such that its P-set $P(M)$ is $X$.
Since $(u, v_u^-) \in X$,
we have $v_u^- \leq_u p_M(u) <_u \top$.
By the above and $v <_u v_u^-$,
we obtain $v <_u p_M(u)$,
which implies $(u, v) \in P(M) = X$.

(ii).
For each $X = P(M) \in \cR_{\mathbf{I}}$ and $u \in U$,
the pair $(u, v) \in U \times V$ belongs to $X$ if and only if $v \leq_u p_M(u)$
by the definition of $P(M)$.
Since $v_u^-$ is the cover of $v$ for $(u,v) \in U \times V$,
we have that
$(u,v) \in X \not\ni (u, v_u^-)$ if and only if $v = p_M(u)$.
Hence, we obtain
$\{ {\sup}_{\hat{r}}(X) \mid X \in \cR_{\mathbf{I}} \} = \left\{\{ (u,v) \in U \times V \mid (u,v) \in X \not\ni (u, v_u^-) \} \mid X \in \cR_{\mathbf{I}} \right\}
= \left\{\{ M \mid X = P(M) \} \mid X \in \cR_{\mathbf{I}} \right\}
= \cS(\mathbf{I})$.
\end{proof}

Thus, by \Cref{thm:optimal-potential-time-complexity}, we obtain the following, which implies \Cref{thm:mincut-stablematching}~(2).

\begin{theorem}
    The $k$-diverse problem with respect to $d_\varphi$ of \StableMatching can be solved
    in $O((|B_k(\varphi)|n^2)^{1 + o(1)}\log (\varphi(k)) \log k + kn)$ time.
    In particular,
    \SumkDiverse of \StableMatching and \CovkDiverse of \StableMatching can be solved
    in $O((kn^2)^{1 + o(1)})$ time
    and
    in $O(n^{2 + o(1)}\log^2 k + kn)$ time,
    respectively.
\end{theorem}
\begin{proof}
    Clearly, we have $|E| = n^2$.
    Since we can construct the partition $\Pi(\cR_\mathbf{I})$ and a DAG with $O(n^2)$ arcs representing $(\Pi^*(\cR_\mathbf{I}), \preceq)$ in $O(n^2)$ time,
    we have $T_P(\mathbf{I}) + T_r(\mathbf{I}) = O(n^2)$ by \Cref{lem:pre-reduction}
    and
    $|A_P| = O(n^2)$.
    Thus, by \Cref{thm:optimal-potential-time-complexity},
    we obtain
    $O((|B_k(\varphi)|n^2)^{1 + o(1)}\log (\varphi(k)) \log k + kn)$ time
    for the $k$-diverse problem with respect to $d_\varphi$,
    $O((kn^2)^{1 + o(1)})$ time for \SumkDiverse,
    and $O(n^{2 + o(1)}\log^2 k + kn)$ time for \CovkDiverse.
\end{proof}

\subsection{The \texorpdfstring{$k$}{k}-diverse problem on the product of total orders}\label{subsec:k-diverse-totalorder}
Very recently, De Berg, Mart\'{i}nez, and Spieksma~\cite{deBerg2025} introduce a framework for solving the $k$-diverse problem on the product of total orders,
which leads to the polynomial-time solvability of \SumkDiverse/\CovkDiverse of \Mincut/\StableMatching.
In this subsection,
we briefly introduce their framework
and see that our framework can capture theirs for \SumkDiverse and \CovkDiverse.
In the following, we do \emph{not} assume that {\sc Prob} has the properties~\eqref{cond:S} and~\eqref{cond:R}.

We first introduce some terminology.
For a distributive lattice $(\cL, \preceq)$,
an element $x \in \cL$ is said to be \emph{join-irreducible}
if $x \neq y \vee z$ for any $y,z \in \cL \setminus \{x\}$.
Let $\cL_{\textup{ir}}$ denote the set of join-irreducible elements in $\cL$.
Then,
it is known~\cite{birkhoff1937} that
the subposet $(\cL_{\textup{ir}}, \preceq)$ of $\cL$ induced by $\cL_{\textup{ir}}$
forms a Birkhoff representation of $\cL$;
the map $I \mapsto \bigvee_{x \in I} x$
is an isomorphism from $(\cI(\cL_{\textup{ir}}), \subseteq)$ to $(\cL, \preceq)$.  
We refer to this Birkhoff representation $(\cL_{\textup{ir}}, \preceq)$ as
the \emph{join-irreducible representation} of $\cL$.
For total orders $(E_1, \leq_1), (E_2, \leq_2), \dots, (E_q, \leq_q)$,
their \emph{product} $\cE = (\cE, \leq)$ is the poset such that
its ground set $\cE$ is the product $E_1 \times E_2 \times \dots \times E_q$ of $E_1, E_2, \dots, E_q$
and the partial order $\leq$ is defined by setting 
$(e_1, e_2, \dots, e_q) \leq (e_1', e_2', \dots, e_q')$ if and only if $e_i \leq_i e_i'$ for all $i \in [n]$.
Actually, this $\cE$ forms a lattice; the meet $(e_1, e_2, \dots, e_q) \wedge (e_1', e_2', \dots, e_q')$ is given by $(\min \{ e_1, e_1' \}, \min \{ e_2, e_2' \}, \dots, \min \{ e_q, e_q' \})$
and the join $(e_1, e_2, \dots, e_q) \vee (e_1', e_2', \dots, e_q')$ is given by $(\max \{ e_1, e_1' \}, \max \{ e_2, e_2' \}, \dots, \max \{ e_q, e_q' \})$,
where
\begin{align}
    \min \{ e_i, e_i' \} \coloneqq
    \begin{cases}
        e_i & \text{if $e_i \leq_i e_i'$},\\
        e_i' & \text{if $e_i' <_i e_i$},
    \end{cases}
    \qquad
    \max \{ e_i, e_i' \} \coloneqq
    \begin{cases}
        e_i' & \text{if $e_i \leq_i e_i'$},\\
        e_i & \text{if $e_i' <_i e_i$}
    \end{cases}
\end{align}
for each $i \in [n]$.
We say that
a subset $\cL \subseteq \cE$ is a \emph{sublattice}
if $\mathcal{L}$ is closed under the meet $\wedge$ and join $\vee$,
i.e., $x, y \in \cL$ implies $x \wedge y, x \vee y \in \cL$.
We can observe that a sublattice of the product of total orders is distributive (see the paragraph after \Cref{thm:k-diverse-totalorder}).

We are ready to introduce the framework of De Berg, Mart\'{i}nez, and Spieksma in~\cite{deBerg2025}.
They impose the following property on a combinatorial problem {\sc Prob}:
\begin{enumerate}[label=(\condT),ref=\condT]
\item \label{cond:T} For any instance $\mathbf{I}$ of {\sc Prob}, there are $q$ total orders $(E_1, \leq_1), (E_2, \leq_2), \dots, (E_q, \leq_q)$ such that the family $\cS(\mathbf{I})$ of solutions of $\mathbf{I}$ is a sublattice of the product of those total orders.
\end{enumerate}
Then they show that,
if {\sc Prob} has the property~\eqref{cond:T}
and we can construct (a DAG representing) the join-irreducible representation $(\cS(\mathbf{I})_{\textup{ir}}, \leq)$ of $\cS(\mathbf{I})$ in polynomial time,
then we can solve \SumkDiverse/\CovkDiverse of {\sc Prob} in polynomial time by using an algorithm for SFM over the distributive lattice $\cI(\cS(\mathbf{I})_{\textup{ir}})$.
The problem \StableMatching has the property~\eqref{cond:T}.
Indeed, for an instance $\mathbf{I} \coloneqq (U, V; (\leq_u)_{u \in U}, (\leq_v)_{v \in V})$ of \StableMatching,
we define $E_u \coloneqq \{ (u,v) \mid v \in V \}$ for each $u\in U$ and
naturally extend the total order $\leq_u$ on $V$ to that on $E_u$ by setting $(u,v) \leq_u (u,v')$ if and only if $v \leq_u v'$.
Then, $\cS(\mathbf{I})$ forms a sublattice of the product of $(E_u, \leq_u)$ for all $u \in U$.
Similarly, \Mincut also has~\eqref{cond:T} by introducing the \emph{left-right order} to $q$ arc disjoint paths,
where $q$ denotes the size of a minimum $s$-$t$ cut (or the maximum number of arc disjoint $s$-$t$ paths); see~\cite{deBerg2023,deBerg2025} for details.

Our aim of this subsection is to show the following:

\begin{theorem}\label{thm:k-diverse-totalorder}
    If a combinatorial problem {\sc Prob} has the property~\eqref{cond:T}, then it also has the properties~\eqref{cond:S} and~\eqref{cond:R}.
\end{theorem}

Since the argument is almost the same as that in \Cref{subsec:k-diverse-stable-matching},
we only give a sketch of the proof of \Cref{thm:k-diverse-totalorder}.
Here,
for a solution $X = (e_1, e_2, \dots, e_q) \in \cS(\mathbf{I})$,
we define its \emph{P-set} $P(X)$ by $P(X) \coloneqq \{ (e_1', e_2', \dots, e_q') \mid e_i' \leq_i e_i \ (i \in [q]) \}$
and denote by $\cR_{\mathbf{I}}$ the family of all P-sets of $\mathbf{I}$,
as in \Cref{subsec:k-diverse-stable-matching}.
Then we can easily see that $(\cS(\mathbf{I}), \leq)$ and $(\cR_\mathbf{I}, \subseteq)$ are isomorphic, both of which are distributive lattices.
Let $E$ be the disjoint union (or direct sum) of $E_1, E_2, \dots, E_q$.
\begin{proof}[Sketch of proof]
    Suppose that {\sc Prob} has the property~\eqref{cond:T}.
Regarding a solution $(e_1, e_2, \dots, e_q) \in \cS(\mathbf{I})$ of $\mathbf{I}$
as a subset $\{ e_1, e_2, \dots, e_q \}$ of $E$,
we see $\cS(\mathbf{I}) \subseteq \binom{E}{q}$,
which implies that {\sc Prob} has the property~\eqref{cond:S}.
We then see the property~\eqref{cond:R}.
We add a new element $\top$ to the ground set $E$ so that the maximal set in $\cR_\mathbf{I}$ is proper in $E$,
and extend the definition of $\leq_i$ by setting $e_i < \top$ for all $e_i \in E_i$.
Now we regard $E \cup \{ \top \}$ as the ground set.
In the following, we construct a pre-reduction map.
For each $e_i \in E_i$,
we denote by $e_i^-$ the element in $E_i \cup \{ \top \}$ that is the cover of $v$ with respect to $\leq_i$,
i.e., $e_i <_i e_i^-$ and there is no element $e_i'$ in $E_i \cup \{ \top \}$ with $e_i <_i e_i' <_i e_i^-$.
We define a map $\hat{r} \colon E \to \left(E \cup \{ \top \}\right)^2$
by $\hat{r}(e) \coloneqq \left(e, e^-\right)$.
Then, by the same argument as in the proof of \Cref{lem:stablematching-pre-reduction},
this $\hat{r}$ is a pre-reduction map.
Hence,
{\sc Prob} has the property~\eqref{cond:R} by \Cref{prop:pre-reduction}.
\end{proof}

Let us see the result of De Berg, Mart\'{i}nez, and Spieksma in~\cite{deBerg2025}
in detail
to compare the time complexity of their algorithm with ours.
We denote by $T_{\textup{ir}}(\mathbf{I})$
the time required to construct a DAG representing the join-irreducible representation $(\cS(\mathbf{I})_{\textup{ir}}, \leq)$,
and by $T_{\textup{SFM}}(n, \textup{EO})$
the time required to minimize a submodular function $f$ such that $f$ has $n$ variables and one value evaluation of $f$ takes \textup{EO} time.
Then,
De Berg, Mart\'{i}nez, and Spieksma
show that,
we can solve
\SumkDiverse/\CovkDiverse of {\sc Prob} having the property~\eqref{cond:T} in $O(T_{\textup{ir}}(\mathbf{I}) + T_{\textup{SFM}}(k|E|, k^2 |E| q))$ time.
Using the state-of-the-art algorithm for SFM given in~\cite{Jiang2021},
we have $T_{\textup{SFM}}(n, \textup{EO}) = O(n^3 \textup{EO})$.
Hence, the running-time of the algorithms of De Berg, Mart\'{i}nez, and Spieksma is
$O(T_{\textup{ir}}(\mathbf{I}) + k^5|E|^4q)$.

Our framework provides much faster algorithms for \SumkDiverse and \CovkDiverse of {\sc Prob} as follows.

\begin{theorem}
    Suppose that a combinatorial problem {\sc Prob} has the property~\eqref{cond:T}.
    Then we can solve \SumkDiverse and \CovkDiverse of {\sc Prob}
    in $O(T_{\textup{ir}}(\mathbf{I}) + |E|q + (|A| + k|E|)^{1+o(1)})$ time
    and
    in $O(T_{\textup{ir}}(\mathbf{I}) + (|E| + k)q + (|A| + |E|)^{1+o(1)} \log^2 k)$ time
    for an instance $\mathbf{I}$ of {\sc Prob},
    respectively,
    where $A$ denotes the arc set of a constructed DAG representing $(\cS(\mathbf{I})_{\textup{ir}}, \leq)$.
\end{theorem}
\begin{proof}
Let $D = (\cS(\mathbf{I})_{\textup{ir}}, A)$ be a DAG representing the join-irreducible representation $(\cS(\mathbf{I})_{\textup{ir}}, \leq)$ of $(\cS(\mathbf{I}), \leq)$ (that can be constructed in $T_{\textup{ir}}(\mathbf{I})$ time).
Since $(\cS(\mathbf{I}), \leq)$ is isomorphic to $(\cR_{\mathbf{I}}, \subseteq)$,
this $D$ also represents \emph{a} Birkhoff representation of $(\cR_{\mathbf{I}}, \subseteq)$. 
In order to apply our framework,
it suffices to obtain a DAG that is \emph{the} Birkhoff representation $(\Pi^*(\cR_{\mathbf{I}}), \preceq)$ of $(\cR_{\mathbf{I}}, \subseteq)$,
or equivalently, obtain the isomorphism between $(\cS(\mathbf{I})_{\textup{ir}}, \leq)$ and $(\Pi^*(\cR_{\mathbf{I}}), \preceq)$,
which enable us to construct a pre-reduction map as in the proof (sketch) of \Cref{thm:k-diverse-totalorder}.

To this end,
we use an arbitrary topological order $\preceq$ of $D$.
Suppose that $X_1 \prec X_2 \prec \cdots \prec X_m$,
where $\{X_1, X_2, \dots, X_m\} = \cS(\mathbf{I})_{\textup{ir}}$.
This can be obtained in linear time in the size of $D$,
which is upper-bounded by $O(T_{\textup{ir}}(\mathbf{I}))$.
Let $Y_j \coloneqq X_1 \vee \cdots \vee X_j$ for each $j \in [m]$.
Since each subset $\{X_1, X_2, \dots, X_j\}$ forms an ideal of $\cS(\mathbf{I})_{\textup{ir}}$,
and the map $I \mapsto \bigvee_{X \in I} X$ gives an isomorphism between $(\cI(\cS(\mathbf{I})_{\textup{ir}}), \subseteq)$
and
$(\cS(\mathbf{I}), \leq)$,
the sequence
$Y_1 < Y_2 < \cdots < Y_m$ forms a maximal chain in $(\cS(\mathbf{I}), \leq)$.
Furthermore,
by the isomorphism $X \mapsto P(X)$ between $(\cS(\mathbf{I}), \leq)$ and $(\cR_{\mathbf{I}}, \subseteq)$,
the sequence
$P(Y_1) \subsetneq P(Y_2) \subsetneq \cdots \subsetneq P(Y_m)$ forms a maximal chain in $(\cR_{\mathbf{I}}, \subseteq)$.
We denote each $X_j$ and $Y_j$ as tuples $X_j = (e_1^j, e_2^j, \dots, e_q^j)$ and $Y_j = (f_1^j, f_2^j, \dots, f_q^j)$, respectively, where $e_i^j$ and $f_i^j$ represent the $i$-th elements of $X_j$ and $Y_j$.
Since $Y_{j+1} = (f_1^{j+1}, \dots, f_q^{j+1}) = (\max\{ f_1^{j}, e_1^{j+1} \}, \dots, \max\{ f_q^{j}, e_q^{j+1} \})$,
we can compute $Y_{j+1}$ from $Y_j$ in $q$ time.
Hence, we can obtain $Y_1, Y_2, \dots, Y_m$ in $mq = O(|E|q)$ time.
Therefore,
we can compute in $O(|E|q) + |E| = O(|E|q)$ time
$P(Y_{j+1}) \setminus P(Y_j) = \{ (e_1, e_2, \dots, e_q) \mid f_i^{j-1} < e_i \leq f_i^j \ (i \in [q]) \} \in \Pi^*(\cR_{\mathbf{I}})$
for all $j \in [m-1]$.
Thus, we have $T_P(\mathbf{I}) + T_r(\mathbf{I}) = T_{\textup{ir}}(\mathbf{I}) + O(|E|q)$.
By \Cref{thm:optimal-potential-time-complexity},
the problems
\SumkDiverse of {\sc Prob} and \CovkDiverse of {\sc Prob} can be solved in
$O(T_{\textup{ir}}(\mathbf{I}) + |E|q + (|A| + k|E|)^{1+o(1)})$ time
and
in $O(T_{\textup{ir}}(\mathbf{I}) + (|E| + k)q + (|A| + |E|)^{1+o(1)} \log^2 k)$ time,
respectively.
\end{proof}

Since $|\cS(\mathbf{I})_{\textup{ir}}|$
is upper-bounded by the length of a maximal chain in the product of total orders $(E_1, \leq_1), (E_2, \leq_2), \dots, (E_q, \leq_q)$,
we have $|\cS(\mathbf{I})_{\textup{ir}}| = O(|E|)$.
Thus, the size $|A|$ of arc set is upper-bounded by $|\cS(\mathbf{I})_{\textup{ir}}|(|\cS(\mathbf{I})_{\textup{ir}}|-1)$,
which is $O(|E|^2)$.
Even in the worst case $|A| = \Theta(|E|^2)$,
the running-time of our algorithms for \SumkDiverse and \CovkDiverse of {\sc Prob}
are 
$O(T_{\textup{ir}}(\mathbf{I}) + (|E|^2 +k|E|)^{1+o(1)})$ time
and $O(T_{\textup{ir}}(\mathbf{I}) + kq + |E|^{2+o(1)} \log^2 k)$ time,
respectively,
which are much faster than the previous ones that take $O(T_{\textup{ir}}(\mathbf{I}) + k^5|E|^4q)$ time.

We also mention that
De Berg, Mart\'{i}nez, and Spieksma~\cite{deBerg2025} introduce another diversity measure $d_{\textup{abs}}$ for {\sc Prob} having the property~\eqref{cond:T}, which is defined by
\begin{align}
    d_{\textup{abs}}(X_1, X_2, \dots, X_k) \coloneqq \sum_{1 \leq i < j \leq k} \|X_i, X_j\|,
\end{align}
where $\|(e_1, e_2, \dots, e_q), (e_1', e_2', \dots, e_q')\|$ denotes the sum of the length of the (unique) maximal chain between $e_\ell$ and $e_\ell'$ in $E_\ell$
over $\ell = 1,2, \dots, q$.
They show the polynomial-time solvability of
the $k$-diverse problem with respect to $d_{\textup{abs}}$ of 
\Mincut/\StableMatching.
Our framework cannot deal with the measure $d_{\textup{abs}}$,
since the value of $d_{\textup{abs}}$ depends not only on the multiplicity $\mu_e(\mathbf{S})$ of solutions $\mathbf{S}$.
It could be an interesting future work to generalize our framework so that it can also treat the diversity measure $d_{\textup{abs}}$.

\section*{Acknowledgments}
We are grateful to Yasunori Kinoshita for his insightful suggestion on the start of this work.
We also thank Yasuaki Kobayashi and Yutaro Yamaguchi for bibliographic information.
This work was partially supported by 
JST ERATO Grant Number JPMJER2301, 
JST ASPIRE Grant Number JPMJAP2302, 
and 
JSPS KAKENHI Grant Numbers 
JP21K17708, 
JP21H03397, 
JP22K17854, 
JP24K02901, 
JP24K21315, 
JP25K00137. 

\bibliographystyle{plain}
\bibliography{arXiv}

\end{document}